%% file: main.tex
\documentclass[conference]{IEEEtran}
\IEEEoverridecommandlockouts
\IEEEaftertitletext{\vspace{-1\baselineskip}}

\AtBeginDocument{%
  \providecommand\BibTeX{{%
    \normalfont B\kern-0.5em{\scshape i\kern-0.25em b}\kern-0.8em\TeX}}}

\input{setting.tex}

\title{Cost-Effective Algorithms for Average-Case Interactive Graph Search\thanks{A short version of the paper will appear in the 38th IEEE International Conference on Data Engineering (ICDE '22), May 9--12, 2022.}}

\makeatletter
\newcommand{\linebreakand}{%
\end{@IEEEauthorhalign}
\hfill\mbox{}\par
\mbox{}\hfill\begin{@IEEEauthorhalign}
}
\makeatother

\author{
    \IEEEauthorblockN{Qianhao Cong}
	\IEEEauthorblockA{Dept.\ of Ind.\ Syst.\ Engg.\ \& Mgmt.\\
		National University of Singapore\\
		cong\_qianhao@u.nus.edu}
	\and
	\IEEEauthorblockN{Jing Tang\IEEEauthorrefmark{1}\thanks{\IEEEauthorrefmark{1}Corresponding author: Jing Tang.}}
	\IEEEauthorblockA{Data Science and Analytics Thrust\\
		The Hong Kong Uni.\ of Sci.\ and Tech.\\
		jingtang@ust.hk}
	\and
	\IEEEauthorblockN{Yuming Huang}
	\IEEEauthorblockA{Dept.\ of Ind.\ Syst.\ Engg.\ \& Mgmt.\\
		National University of Singapore\\
		huangyuming@u.nus.edu}
	\linebreakand
	\IEEEauthorblockN{Lei Chen}
	\IEEEauthorblockA{Data Science and Analytics Thrust\\
        The Hong Kong Uni.\ of Sci.\ and Tech.\\
		leichen@ust.hk}
	\and
	\IEEEauthorblockN{Yeow Meng Chee}
	\IEEEauthorblockA{Dept.\ of Ind.\ Syst.\ Engg.\ \& Mgmt.\\
		National University of Singapore\\
		ymchee@nus.edu.sg}
}

\begin{document}

\maketitle

\thispagestyle{plain}
\pagestyle{plain}

\begin{sloppy}
\input{chapters/00-Abstract}
\input{chapters/01-Introduction}
\input{chapters/03-Preliminary}
\input{chapters/04-AnalysisOnGreedyAlgorithm}
\input{chapters/05-Algorithm}
\input{chapters/06-Experiment}
\input{chapters/02-RelatedWork}

\input{chapters/Conclusion}

\section*{Acknowledgment}
Jing Tang's work is partially supported by HKUST(GZ) under a Startup Grant. Lei Chen's work is partially supported by National Key Research and Development Program of China Grant No.\ 2018AAA0101100, the Hong Kong RGC GRF Project 16207617, CRF Project C6030-18G, C1031-18G, C5026-18G, RIF Project R6020-19, AOE Project AoE/E-603/18, Theme-based project TRS T41-603/20R, China NSFC No.\ 61729201, Guangdong Basic and Applied Basic Research Foundation 2019B151530001, Hong Kong ITC ITF grants ITS/044/18FX and ITS/470/18FX, Microsoft Research Asia Collaborative Research Grant, HKUST-NAVER/LINE AI Lab, HKUST-Webank joint research lab grants.

\end{sloppy}

\balance
\bibliographystyle{abbrvnat}
\bibliography{reference}


\end{document}

%% file: setting.tex
\usepackage[numbers,sort&compress]{natbib}
\usepackage[utf8]{inputenc} 
\usepackage[T1]{fontenc}    
\usepackage{microtype}      
\usepackage{enumitem}
\usepackage{booktabs}
\usepackage{graphicx}
\usepackage{amsmath,amssymb,amsfonts,amsthm}
\usepackage{url}
\usepackage{balance}
\usepackage{color}
\usepackage[caption=false,labelformat=simple]{subfig}
\usepackage[linesnumbered,ruled,vlined,norelsize]{algorithm2e}

\SetCommentSty{textsf}

\usepackage{xcolor}
\usepackage[unicode]{hyperref}
\hypersetup{
	pdfstartview=FitH,
	bookmarks=true,
	bookmarksnumbered=true,
	bookmarksopen=true,
	colorlinks,
	linkcolor={violet!90!black},
	citecolor={blue!50!black},
	urlcolor={blue!80!black}
}

\newtheorem{theorem}{Theorem}
\newtheorem{example}{Example}
\newtheorem{lemma}{Lemma}
\newtheorem{definition}{Definition}

\newcommand{\abs}[1]{\lvert#1\rvert}
\newcommand{\spara}[1]{\vspace{2mm}\noindent\textbf{#1.}}

\newcommand{\topdown}{\textsf{TopDown}\xspace}
\newcommand{\igs}{\textsf{WIGS}\xspace}
\newcommand{\migs}{\textsf{MIGS}\xspace}
\newcommand{\MIGS}{\textsf{MIGS}\xspace}
\newcommand{\IGS}{\textsf{FrameworkIGS}\xspace}

\newcommand{\greedyN}{\textsf{GreedyNaive}\xspace}
\newcommand{\greedyT}{\textsf{GreedyTree}\xspace}
\newcommand{\greedyG}{\textsf{GreedyDAG}\xspace}
\newcommand{\getp}{\textsf{GetReachableSetWeight}\xspace}
\newcommand{\getpDFS}{\textsf{SetWeightDFS}\xspace}
\newcommand{\adjust}{\textsf{AdjustWeight}\xspace}
\newcommand{\done}{\hspace*{\fill} {$\square$}}

\DeclareMathOperator*{\argmax}{\mathop{\arg\max}}
\DeclareMathOperator*{\argmin}{\mathop{\arg\min}}

\newcommand{\compilehidecomments}{false}
\ifthenelse{ \equal{\compilehidecomments}{true} }{%
	\newcommand{\jing}[1]{}
}{
	\newcommand{\jing}[1]{{\color{blue}  [\text{Jing:} #1]}}
}
\ifthenelse{ \equal{\compilehidecomments}{true} }{%
	\newcommand{\cqh}[1]{}
}{
	\newcommand{\cqh}[1]{{\color{blue}  [\text{cqh:} #1]}}
}

%% file: chapters/00-Abstract.tex
\begin{abstract}
Interactive graph search (IGS) uses human intelligence to locate the target node in hierarchy, which can be applied for image classification, product categorization and searching a database. Specifically, IGS aims to categorize an object from a given category hierarchy via several rounds of interactive queries. In each round of query, the search algorithm picks a category and receives a boolean answer on whether the object is under the chosen category. The main efficiency goal asks for the minimum number of queries to identify the correct hierarchical category for the object. In this paper, we study the {\em average-case interactive graph search (AIGS)} problem that aims to minimize the expected number of queries when the objects follow a probability distribution. We propose a greedy search policy that splits the candidate categories as evenly as possible with respect to the probability weights, which offers an approximation guarantee of $O(\log n)$ for AIGS given the category hierarchy is a directed acyclic graph (DAG), where $n$ is the total number of categories. Meanwhile, if the input hierarchy is a tree, we show that a constant approximation factor of ${(1+\sqrt{5})}/{2}$ can be achieved. Furthermore, we present efficient implementations of the greedy policy, namely \greedyT and \greedyG, that can quickly categorize the object in practice. Extensive experiments in real-world scenarios are carried out to demonstrate the superiority of our proposed methods.
\end{abstract}

%% file: chapters/01-Introduction.tex
\section{Introduction}
Crowdsourcing services, such as Amazon's Mechanical Turk and CrowdFlower, allow users to generate tasks for crowd workers to complete in exchange for rewards. The main goal of crowdsourcing is to solve problems that are hard for computers. During the past decades, crowdsourcing services have emerged as a powerful mechanism to process data for many tasks, such as object categorization~\cite{parameswaran2011human,li2020efficient}, data-labeling and dataset creation~\cite{deng2009imagenet,lin2014microsoft}, entity resolution~\cite{wang2012crowder,vesdapunt2014crowdsourcing}, data filtering~\cite{parameswaran2012crowdscreen,parameswaran2014optimal}, and SQL-like query processing~\cite{li2017cdb,karger2011human,davidson2013using,guo2012so,venetis2012max}.

Recently, \citet{Tao:IGS} proposed the interactive graph search (IGS) problem that aims to find an initially unknown target node on a direct acyclic graph (DAG) by asking questions like ``is the target node reachable from a given node $u$?''.  In the beginning, all nodes on the DAG are considered as candidates. By sequentially asking questions, we can gradually narrow down the candidates and finally identify the target node. In real-world scenarios, these questions may be answered through a crowdsourcing platform. Meanwhile, crowdsourcing platforms usually charge a flat fee for each question, i.e., a unit cost per question. Therefore, a natural aim is to reduce the number of queries needed for finding the target node. There are many real-world applications~\cite{parameswaran2011human} that can be characterized by the IGS problem, such as image classification~\cite{Tao:IGS}, product categorization~\cite{li2020efficient}, debugging of workflows~\cite{heinis2008efficient,cohen2006towards}, etc. As an immediate illustration, the following example shows the task of labeling an image by crowdsourcing~\cite{parameswaran2011human}.

\begin{figure}[!tbp]
	\centering
	\includegraphics[width=0.45\textwidth]{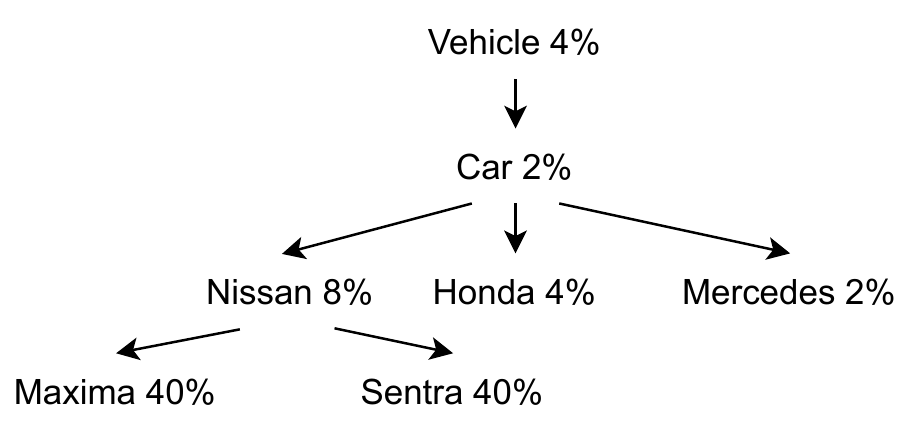}
	\vspace{-3mm}
	\caption{Image categorization~\cite{parameswaran2011human}.}
	\label{fig:ExampleCategoryIdentification}
\end{figure}

\begin{example}[Image Categorization]\label{exp:ImageIdentification}
In (supervised) machine learning, algorithms usually require labeled data to train the models. Image category identification via crowdsourcing is an important channel to obtain such labeled data. In particular, humans are hired to label images according to a given hierarchy by answering a sequence of questions. 
Consider the scenario that we attempt to label a vehicle image given a hierarchy shown in \figurename~\ref{fig:ExampleCategoryIdentification}, where the proportions of different vehicle categories are different. Note that the target can be either a leaf node or a non-leaf node in the hierarchy. For example, with a probability of $40\%$, the image shows a \textit{Sentra}, while with another probability of $8\%$, the image is a \textit{Nissan} but neither a \textit{Maxima} nor a \textit{Sentra}.

We may first ask a question like ``is this a \textit{car}?''. If the answer is \textit{no}, we can immediately label the image as a non-car \textit{vehicle}. If the answer is \textit{yes}, we continue the search process in the tree rooted at the node of \textit{car}. Asking sequential questions in a similar way, we can finally find the most suitable label to describe the image. Suppose that we receive a \textit{yes}-answer to both questions of ``is this a \textit{car}?'' and ``is this a \textit{Honda}?''. Then, a label of \textit{Honda} is placed on the image.\done
\end{example}

A straightforward solution is a \topdown method that asks questions from the root of the hierarchy. Specifically, \topdown first queries every child of the root until it gets a \textit{yes}-answer. If no such a child exists, the root will be returned as the target node. Otherwise, such a child will be set as the new root to repeat the process. Take \figurename~\ref{fig:ExampleCategoryIdentification} as an illustration and assume that \textit{Sentra} is the target node. The \topdown algorithm first asks ``is this a \textit{car}?'' and receives a \textit{yes}-answer. Then, it asks ``is this a \textit{Nissan}?'' and again gets a \textit{yes}-answer. It further asks ``is this a \textit{Maxima}?'' and this time a \textit{no}-answer is returned. Finally, it asks ``is this a \textit{Sentra}'' with a \textit{yes}-answer so that the target node is identified successfully. This strategy can correctly find the target node but unfortunately is cost-ineffective. Since we know that most of the images are \textit{Maxima} or \textit{Sentra}, a smarter strategy can query these two nodes first so that with a probability of $80\%$, asking at most two questions can identify the target node. This observation motivates us to develop cost-effective algorithms to minimize the cost.
 
\citet{Tao:IGS} studied the \textit{worst-case} interactive graph search (WIGS) that aims to minimize the \textit{maximum} number of queries required to identify every possible target node and proposed a series of heavy-path-based binary searches with near-optimal theoretical guarantees. However, for the task of image categorization by crowdsourcing, the data owner hires humans to label the images with a desire of minimizing the total cost of labeling all images, rather than minimizing the maximum cost of labeling one image. 
\begin{example}[Query Cost]\label{exp:cost}
    Consider that there are $100$ images with proportions as given in \figurename~\ref{fig:ExampleCategoryIdentification}, e.g., $4$ images are non-car \textit{vehicles} and $40$ images are \textit{Maxima}. An optimal solution to WIGS sequentially queries \textit{Nissan}, \textit{Car}, \textit{Honda}, and \textit{Mercedes} with answers of \textit{no}, \textit{yes}, \textit{no}, and \textit{yes} returned, assuming that the image is a \textit{Mercedes}. In fact, this is the worst case for such a solution requiring $4$ queries. Similarly, the numbers of queries asked are $2$, $4$, $3$, $3$, $2$ and $3$ if the correct categories are \textit{Vehicle}, \textit{Car}, \textit{Honda}, \textit{Nissan}, \textit{Maxima} and \textit{Sentra}, respectively. As a result, the total cost is $260$.
    
    On the other hand, consider an alternative solution that sequentially queries \textit{Maxima}, \textit{Sentra}, \textit{Nissan}, \textit{Car}, \textit{Honda}, and \textit{Mercedes} with a cost of $6$, if the image is again a \textit{Mercedes}. Meanwhile, if the correct categories are \textit{Vehicle}, \textit{Car}, \textit{Honda}, \textit{Nissan}, \textit{Maxima} and \textit{Sentra}, the costs become $4$, $6$, $5$, $3$, $1$ and $2$, respectively. Hence, the total cost is $204$. Apparently, in terms of the worst-case cost, this solution requiring $6$ queries is inferior to the optimal solution to WIGS requiring $4$ queries only, whereas in terms of the average-case cost, the former (with a cost of $2.04$ in average) remarkably outperforms the latter (with a cost of $2.6$ in average).\done
\end{example}


Motivated by the deficiency of existing solutions, in this paper, we focus on the \textit{average-case} interactive graph search (AIGS) problem that attempts to minimize the \textit{expected} number of queries when the target node follows a-priori known probability distribution. It is trivial to see that minimizing the total cost of categorizing a batch of objects is factually equivalent to AIGS, assuming that we know the distribution of the target nodes in this batch of tasks. Note that in some real applications, it may be impossible to get the true data distribution. For such scenarios, we may apply a simple online learning approach that dynamically adjusts the empirical probability distribution on the fly upon obtaining the category result of each object. Usually, the empirical statistics can well characterize the real data distribution as long as a sufficiently large number of labels are obtained.

We find that obtaining the optimal solution for AIGS is computationally intractable. To address AIGS, we propose a greedy search policy that picks a query node to split the candidate nodes as evenly as possible, taking into account the probability weight of each node. With a rigorous and thorough analysis, we show that the greedy policy can achieve strong theoretical guarantees, especially a constant factor given the input hierarchy is a tree. In addition, a naive implementation of the greedy policy is to evaluate the splitting size of every candidate node in each round, which is time-consuming. To tackle this issue, we develop efficient algorithms for variant scenarios, namely \greedyT and \greedyG, to accelerate the query node selection according to the greedy policy.

In summary, we make the following contributions.
\begin{itemize}[topsep=2mm, partopsep=0pt, itemsep=1mm, leftmargin=18pt]
	\item We show that computing the optimal policy for AIGS is NP-hard, even if the hierarchy has a tree structure.
	\item We propose a cost-effective approach for AIGS. We show that our greedy policy offers an approximation ratio of $O(\log n)$ for any input DAG with $n$ nodes and a constant factor of $(1+\sqrt{5})/2$ when the input hierarchy is a tree (Section~\ref{sec:analysis}).
	\item We devise efficient implementations of the greedy policy, incorporating several acceleration techniques. In particular, our \greedyT and \greedyG algorithms run in $O(nhd)$ and $O(nm)$ time for tree and DAG hierarchies respectively, improving the naive $O(n^2m)$ time algorithm significantly, where $n$ and $m$ are the number of nodes and edges, $h$ is the length of longest path, and $d$ is the maximum out-degree of nodes in the hierarchy. (Section~\ref{sec:algorithm}).
	\item We conduct extensive experiments with real-world datasets to evaluate our proposed methods, and the experimental results strongly corroborate the effectiveness and efficiency of our approach (Section~\ref{sec:experiment}).
\end{itemize}

%% file: chapters/03-Preliminary.tex
\section{Problem Definition}\label{sec:preliminary}
In this paper, we study the problem of \textit{average-case interactive graph search (AIGS)}. That is, we sequentially ask the crowd some simple reachability questions with boolean answers, i.e.,~\textit{yes} or \textit{no}, to gradually narrow down the potential categories of the object. Meanwhile, the cost of crowdsourcing is usually determined by the number of questions asked. Therefore, our aim is to correctly identify the hierarchical category by asking the minimum \textit{expected} number of questions, i.e., minimizing the \textit{expected} cost. 
For ease of reference, Table~\ref{tab:notations} summarizes the notations that are frequently used.

 \begin{table}[!tbp]
 	\centering
 	\setlength{\tabcolsep}{0.6em} 
 	\caption{Frequently used notations.}
 	\label{tab:notations}
 	\begin{tabular}{cl}
 		\toprule
 		\textbf{Notation} & \textbf{Description}\\
 		\midrule
 		$G=(V,E)$ & a DAG $G$ with node set $V$ and edge set $E$\\
 		$n$ & the number of nodes in $G$, i.e., $n=\lvert V\rvert$\\
 		$m$ & the number of edges in $G$, i.e., $m=\lvert E\rvert$\\
 		$z$ & target node, e.g., the category of an unlabeled object\\
 		$reach(u)$ & query on node $u$, i.e., whether $z$ is reachable from $u$ in $G$\\
 		$G_u$ & the subgraph of $G$ rooted at node $u$\\
 		\bottomrule
 	\end{tabular}
 \end{table}

Specifically, AIGS is formally defined as follows. We abstract a category hierarchy as a directed acyclic graph (DAG) $G=(V,E)$ with a set $V$ of $n$ nodes and a set $E$ of $m$ edges. We assume that there is only one \textit{root} in $G$. If there are multiple roots, we can simply add a dummy node to $G$ with an outgoing edge to every original root, which generates a new DAG with one root only. Given an unknown \textit{target node} (e.g., an unlabeled object) $z\in V$, there is an \textit{oracle} that can answer questions. For any \textit{query node} $q\in V$, the oracle returns a boolean answer, denoted as $reach(q)$, as follows,
\begin{equation*}
	reach(q)=
	\begin{cases}
	yes, &\text{if there is a directed path from $q$ to $z$,}\\
	no, &\text{otherwise.}
	\end{cases}
\end{equation*}
It is easy to see that $reach(q)=yes$ if and only if $q$ can \textit{reach} $z$, which characterizes the \textit{reachability} from $q$ to $z$. For IGS, the algorithm interactively picks a query node $u$ and receives a boolean answer $reach(u)$. The algorithm repeats the process until it determines the target node $z$. 

In particular, in each round of interaction, if it gets a \textit{yes}-answer, the target node must be reachable from the query node $u$. That is, if $reach(u)=yes$, it clearly holds that the target node $z\in G_u$ and the search graph is updated to $G_u$, where $G_u$ is the subgraph of $G$ rooted at the query node $u$. Otherwise, if $reach(u)=no$, $G$ is updated by $G\setminus G_u$ as $z\in G\setminus G_u$. The algorithm stops to return the target node when the search graph just has one node. We present the main procedure of IGS in Algorithm~\ref{alg:framework_of_IGS}, referred to as \IGS. As can be seen, the critical task of the algorithm design for IGS lies in the choice of query node (Line~\ref{alg:framework_of_IGS:select_query_node}). 

We define the \textit{cost} as the number of questions the algorithm asks, assuming that each question charges a fixed price (e.g.,~$\$1$ per question). In fact, as shall be discussed in Section~\ref{section:Analysis:CAIGS}, we can extend to a general scenario where the payment for each question is distinct to reveal the difficulties of questions, e.g., hard questions are more expensive than easy questions. A natural goal of IGS is to identify the target node with the minimum cost. In this paper, we consider that each node $v$ is associated with a probability $p(v)$ measuring the likelihood of $v$ being the target node, e.g., a-priori known data distribution. Given a set $S$ of nodes, let $p(S)=\sum_{v\in S}p(v)$ denote the probability of the target node being one node in $S$, e.g.,~$p(V)=1$. Then, we define the \textit{expected cost} of identifying a target node as the expected number of questions asked by the algorithm, taking into account the randomness of the target node. Therefore, we aim to study the interactive question-asking strategies (also known as policies) for the \textit{average-case interactive graph search (AIGS)} problem with the goal of minimizing the expected cost. For convenience, we abuse notation and use the terms ``policy'' and ``algorithm'' interchangeably. 

 \begin{algorithm}[!tbp]
 \setlength{\hsize}{0.94\linewidth}
 	\caption{$\IGS(G)$}
 	\label{alg:framework_of_IGS}
 	\KwIn{an input DAG $G$}
 	\KwOut{the target node $z$ in $G$}
 	\While{$G$ has at least two nodes}{
 		select a query node $u$ in $G$\;\label{alg:framework_of_IGS:select_query_node}
 		\uIf{$reach(u)=yes$}{
 			$G \gets G_u$\label{alg:framework_of_IGS:yes}\tcp*{$G_u$ is the subgraph of $G$ rooted at $u$}
 		}
 		\Else{
 			$G \gets G \setminus G_u$\;\label{alg:framework_of_IGS:no}
 		}
 		
 	}
 	\Return the only node in $G$\;
 \end{algorithm}
 

\begin{definition}[Average-Case Interactive Graph Search]\label{def:problem-AIGS}
	Given a DAG hierarchy $G=(V,E)$ and an (unknown) target node $z\in V$ following the a-priori known probability distribution $p(\cdot)$, AIGS asks for a query policy to identify the target node such that the expected cost is minimized.
\end{definition}
\vspace{-2mm}
\spara{Remark} In the AIGS problem, we assume a-priori known probability distribution. In practice, the probability distribution can be empirically learned from the historical data, e.g., a simple statistical result of the categorized objects. Moreover, when there is a batch of objects (e.g., a batch of unlabeled images) required to be categorized, these objects generally follow the probability distribution. As a result, minimizing the total cost of categorizing a batch of objects, which is the average cost of each object multiplying the number of objects, is equivalent to the optimization problem of AIGS.

%% file: chapters/04-AnalysisOnGreedyAlgorithm.tex
\section{Greedy Policy for AIGS}\label{sec:analysis}
In this section, we show that finding the optimal solution for AIGS is computationally intractable and propose a natural greedy policy with provable approximation guarantees.

\subsection{Impossibility Results}
\begin{lemma}\label{the:np_hardness_for_IGS}
The AIGS problem is NP-hard. In fact, there cannot be any $o(\log n)$-approximate algorithm for AIGS unless $\mathrm{NP} \subseteq \mathrm{DTIME}\left(n^{O(\log\log n)}\right)$.
\end{lemma}

To prove Lemma~\ref{the:np_hardness_for_IGS}, we begin by showing that the IGS problem is equivalent to the search in \textit{partially ordered set (poset)} problem~\cite{daskalakis2011sorting}. We first introduce the definition of poset.

\begin{definition}[Partially Ordered Set]
A partially ordered set is a pair $(V,\le_R)$, where $V$ is a set of objects and $\le_R$ is a partially ordered relation which satisfies the following properties:
\begin{itemize}
    \item Reflexivity: $a \le_R a$.
    \item Antisymmetry: if $a \le_R b$ and $b \le_R a$, then $a=b$.
    \item Transitivity: if $a \le_R b$ and $b \le_R c$, then $a \le_R c$. 
\end{itemize}
\end{definition}

Reflexivity shows that every object is related to itself. Antisymmetry indicates that for any two distinct objects, they cannot relate to each other. Transitivity reflects that the relation is transitive.

\begin{definition}[Search in Poset]
Given an (unknown) target object $z$ in a poset $(V,\le_R)$, the search problem aims to locate the target objects by queries. To search the target object, each query on object $x\in V$ will return a boolean result:
\begin{itemize}
	\item \textit{yes}, if the target object is related to $x$, i.e., $z\le_R x$;
	\item \textit{no}, otherwise.
\end{itemize}
\end{definition}
We make a connection between the IGS problem and the search problem in poset.
\begin{lemma}\label{the:IGS_to_searchinposet_appendix}
	The IGS problem is equivalent to the search problem in a poset.
\end{lemma}
\begin{proof}
	The reachability of IGS satisfies the following properties.
	\begin{itemize}
		\item Reflexivity: Each node can reach itself.
		\item Antisymmetry: For two nodes $u,v$, if $u$ can reach $v$ and $v$ can reach $u$, then $u$ should equal to $v$.
		\item Transitivity: For three nodes $u,v,w$, if $u$ can reach $v$ and $v$ can reach $w$, then $u$ can reach $w$ by taking $v$ as a relay node.
	\end{itemize}
	Hence, the reachability relation in the IGS problem is a partially ordered relation. Moreover, given a target node $z\in V$, the query on node $q\in V$ will get a \textit{yes} answer if and only if $z$ is reachable from $q$, which satisfies query in a poset. Therefore, the IGS problem is a search problem in a poset.
	
	On the other hand, given a poset, we can construct a DAG hierarchy as follows. For any two distinct objects $a$ and $b$ in the poset, if $a\leq_R b$ and there does not exist any $c\in V\setminus\{a,b\}$ such that $a\leq_R c$ and $ c\leq_R b$, i.e.,~$a$ is directly related to $b$, then $a$ is the child of $b$ in the hierarchy. It is easy to verify that the constructed hierarchy is a DAG and for two objects $a$ and $b$ satisfying $a\leq_R b$, it must hold that $a$ is reachable from $b$ in the hierarchy. This implies that searching a poset can be represented by an IGS problem.
\end{proof}

Now we are ready to prove Lemma~\ref{the:np_hardness_for_IGS}.

\begin{proof}[Proof of Lemma~\ref{the:np_hardness_for_IGS}]
    \citet{cicalese2011complexity} proved that deriving the optimal solution for the average-case poset search is NP-hard even if the poset has a tree structure (i.e.,~the input hierarchy of AIGS is a tree). Moreover, \citet{cicalese2011complexity} pointed out that for the problem of minimizing the weighted average number of queries to identify an initially unknown object for general posets, there cannot be any $o(\log n)$-approximate algorithm unless $\mathrm{NP} \subseteq \mathrm{DTIME}\left(n^{O(\log\log n)}\right)$, where $n$ is the number of nodes. Combining with Lemma~\ref{the:IGS_to_searchinposet_appendix}, we complete the proof.
\end{proof}

Lemma~\ref{the:np_hardness_for_IGS} implies that in general, it is impossible to construct the optimal solution for AIGS in polynomial time unless $\mathrm{P}\!=\!\mathrm{NP}$. Even worse, devising a good approximate algorithm for AIGS within a factor of $o(\log n)$ is also impossible. In the following, we propose a greedy algorithm that can provide provable theoretical guarantees.

\begin{figure*}[!hbpt]
	\centering
	\subfloat[An Input Hierarchy]{\includegraphics[height=0.145\textheight]{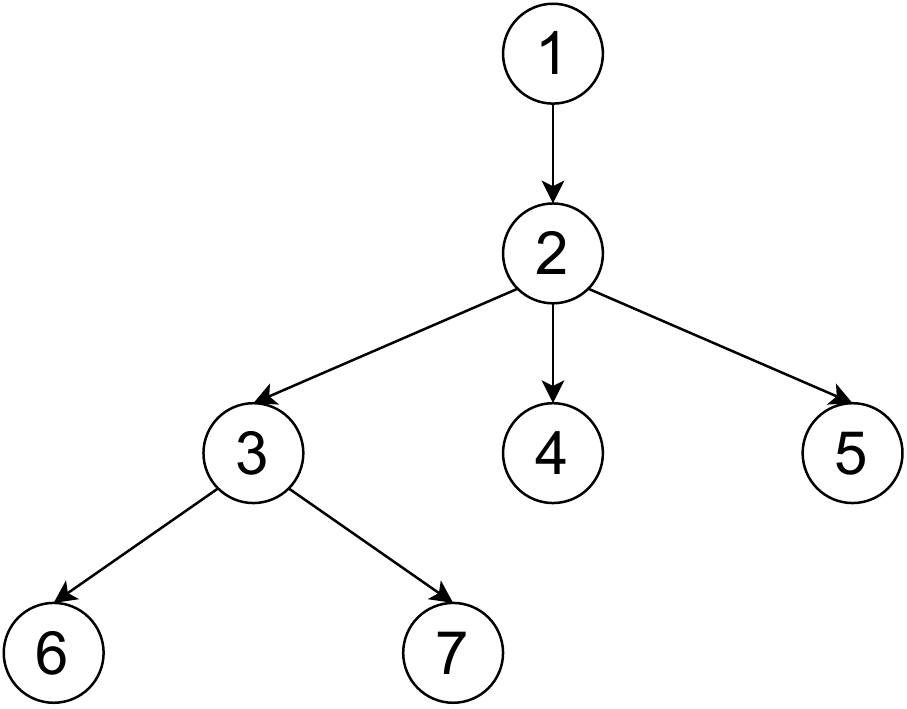}\label{fig:IGS-GraphExample}}\hfil
	\subfloat[Decision Tree]{\includegraphics[height=0.145\textheight]{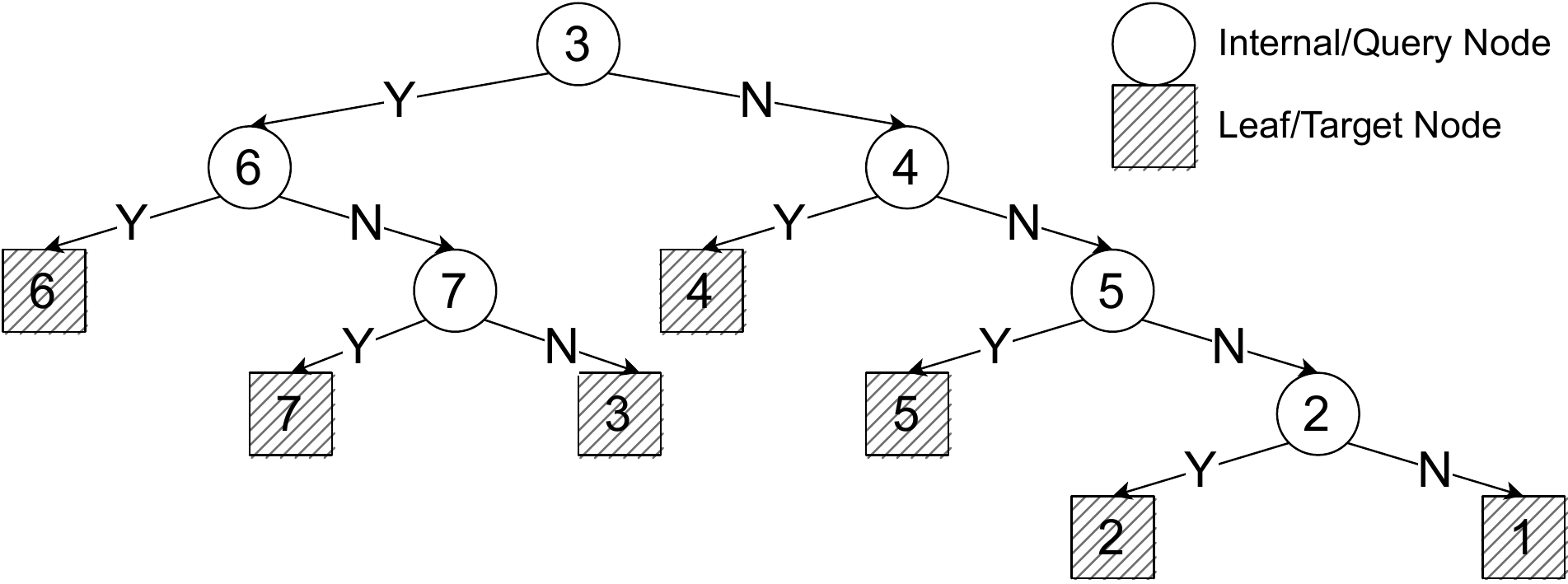}\label{fig:IGS-DTExample}}
	\caption{Decision tree for the query policy of interactive graph search.}\label{fig:IGS-DT}
\end{figure*}

\subsection{Greedy Policy}
Recall that a key step of AIGS is to select the query node $u$ which will split the candidate DAG $G$ into two DAGs $G_u$ and $G \setminus G_u$, where $G_u$ is the subgraph of $G$ rooted at $u$ with all its descendants. Then, the target node $z$ is either in $G_u$ or in $G\setminus G_u$. Hence, to minimize the expected cost, intuitively, a ``good'' strategy is to filter nodes with total weight as large as possible with respect to the probability distribution. That is, the query splits the candidate set of nodes into two parts with the most equal probability weights. This strategy can avoid asking a large number of questions to identify the nodes with high probability. In this paper, we utilize such a natural greedy policy to address the AIGS problem. Specifically, the greedy policy picks the node $u$, referred to as the \textit{middle point}, such that the difference between $p(G \setminus G_u)$ and $p(G_u)$ is minimized.
\begin{definition}[Middle Point]\label{def:middle_point}
	Given a DAG $G$ associated with a weight $p(v)$ for each node $v$, the middle point $u^\ast$ of $G$ is defined as
	\begin{equation*}
		u^\ast:=\argmin\limits_{u \in V}\ \abs{p(G_u)-p(G \setminus G_u)}=\argmin\limits_{u \in V}\ \abs{2p(G_u)-p(G)}.
	\end{equation*}
\end{definition}
Note that there may exist multiple middle points in a DAG. If such a case happens, the greedy policy just arbitrarily chooses one of them.

\subsection{Theoretical Guarantees}\label{sec:analysis:guarantees}
In what follows, we show that the greedy policy can achieve an approximation guarantee of $O(\log n)$ for AIGS, which matches the best achievable ratio. A key step of our analysis is to map AIGS to the binary decision tree problem~\cite{garey1972optimal,chakaravarthy2007decision}.

There has been extensive research~\cite{Kosaraju1999optimal,chakaravarthy2007decision} showing that the minimum (probability) weight of nodes will have a significant impact on the approximation ratio for the decision tree problem. Fortunately, a rounding technique~\cite{chakaravarthy2007decision} can be used tackle the negative impact of the minimum weight. In particular, for each node $u$, its weight will be rounded from $p(u)$ to $w(u)$ as follows,
\begin{equation}\label{equation:rounded-greedy}
	w(u)=\left\lceil\frac{n^2\cdot p(u)}{\max_{v \in V} p(v)}\right\rceil.
\end{equation}
Then, the greedy policy is used to construct the decision tree based on the rounded weights $w(\cdot)$, namely rounded greedy policy. \citet{chakaravarthy2007decision} proved that the rounded greedy policy achieves an approximation ratio of $2(1+3\ln n)$.\footnote{The original analysis~\cite{chakaravarthy2007decision} requires the probabilities to be rational numbers, whereas we assume the weights to be real numbers. However, this mismatch does not matter a lot, since the machine uses a limited number of bits to record a real number with sufficient precision.} This approximation result also applies to AIGS.

\begin{theorem}\label{thm:round-greedy}
The rounded greedy algorithm provides an approximation ratio of $2(1+3\ln n)$ for AIGS.
\end{theorem}

To prove Theorem~\ref{thm:round-greedy}, we make the connection between our AIGS problem and the decision tree problem~\cite{chakaravarthy2007decision,garey1972optimal}. In particular, a target node of AIGS corresponds to a leaf node of decision tree while a query node of AIGS corresponds to an internal node of decision tree (i.e., whether the target node is reachable from the query node). Hence, we revisit the AIGS problem from the view of decision tree. In particular, as the answer is a boolean, i.e.,~\textit{yes} or \textit{no}, a binary decision tree is considered in practice, which is formally defined as follows.

\begin{definition}[Binary Decision Tree~\cite{chakaravarthy2007decision}]\label{def:dt}
The binary decision tree problem take as input a table having $N$ rows and $M$ columns. Each row is called an object and columns are binary attributes of these objects. The aim of this problem is to output a binary tree where each leaf is associated with an object and each internal node is associated with a test for one attribute. If an object passes a test, it goes to the left branch; otherwise, it goes to the right branch. Each object can be uniquely identified by a path from the root to its representing leaf. The goal of a decision tree problem is to find such a tree that minimizes the weighted sum of the depths of all the leaves, considering that each object (i.e.,~leaf) is associated with a weight. 
\end{definition}

\begin{lemma}
The AIGS problem can be reduced to the binary decision problem.
\end{lemma}
\begin{proof}
To reduce the AIGS problem to the binary decision tree problem, we can consider each node in the input hierarchy as an object and the reachability relations as the attributes. Specifically, we can transfer the input hierarchy with $n$ nodes to an $n \times n$ table where the attribute of $i$-th row and $j$-th column equals to $1$ if node $i$ is reachable from node $j$, and equals to $0$ otherwise. Then, a query on node $j$ in AIGS is exactly a test for one attribute $j$ in the decision tree. As a result, minimizing the expected cost for AIGS is equivalent to minimizing the weighted sum of the depths of all the leaves for the binary decision tree problem. This completes the proof.
\end{proof}

Based on this observation, we construct a binary decision tree according to the search strategy of AIGS and calculate the cost based on the constructed decision tree.

\begin{definition}[Decision Tree for AIGS]
	We model the query policy of AIGS as a decision tree, where each internal node represents a question (i.e.,~the query node) and each leaf node represents a search result (i.e.,~the target node). Starting from the root of a decision tree, the algorithm asks questions following the below rules.
	\begin{itemize}
		\item If the answer is \textit{yes}, i.e.,~the target node is reachable from the query node, it goes to the left child.
		\item Otherwise, if the answer is \textit{no}, it goes to the right child.
	\end{itemize}
	The above process stops at a leaf node, which will be returned as the search result.
\end{definition}

According to the above definition, it is trivial to see that given any DAG $G$, the size of its decision tree $D$ is at most twice the size of $G$, since every node of $G$ appears as a leaf node of $D$ and meanwhile at most all nodes of $G$ serve as internal nodes of $D$. Furthermore, it is easy to observe that the number of questions asked for identifying a target node is equal to the depth of the corresponding leaf node in $D$, i.e.,~the length of its path from the root in $D$. Meanwhile, the weight of each leaf node $v$ is equal to the probability $p(v)$ of the node $v$ being the target node. Hence, we formally define the cost of a query policy leveraging the notion of decision tree.
\begin{definition}[Cost of Query Policy]\label{def:cost1}
	Given a query policy for the AIGS problem, let $D$ be the decision tree constructed according to the policy. Denote by $L(D)$ the set of all leaf nodes in $D$, and by $\ell(v)$ the depth of every leaf node $v\in L(D)$. Then, the cost of $D$, i.e.,~the expected cost of the query policy, is given by
	\begin{equation}\label{eq:cost1}
	cost(D)=\sum_{v \in L(D)}p(v)\ell(v).
	\end{equation}
\end{definition}

\begin{example}
	\figurename~\ref{fig:IGS-GraphExample} gives the graph representation of the hierarchy given in \figurename~\ref{fig:ExampleCategoryIdentification}. Suppose that the nodes have equal weights, i.e.,~$p(v)=1/7$ for every $v\in V$. \figurename~\ref{fig:IGS-DTExample} shows the decision tree of a (greedy) policy. Assume that node $5$ is the target node. Then, the algorithm will sequentially query on nodes $3$, $4$ and $5$ with answers of \textit{no}, \textit{no} and \textit{yes}, respectively, so that the target node is figured out with no ambiguity. Moreover, the expected cost of this policy is $\frac{1}{7}\times(2\times 2+ 3\times 3 + 2\times 4)=3$.\done
\end{example}

Now, we are ready to prove Theorem~\ref{thm:round-greedy}.
\begin{proof}[Proof of Theorem~\ref{thm:round-greedy}]
    Our above analysis shows that AIGS can be transferred to a special case of average-case binary decision tree. In addition, \citet{chakaravarthy2007decision} proved that the rounded greedy policy achieves an approximation ratio of $2(1+3\ln n)$ for average-case binary decision tree. Putting it together completes the proof.
\end{proof}

Next, we consider two special cases of AIGS: (i) the input hierarchy is a tree, and (ii) the probability of every node being the target is identical, i.e.,~$p(v)=1/n$ for every $v\in V$. We show that the approximation ratios achieved by the greedy policy under the two scenarios are much better than $O(\log n)$ leveraging the theoretical developments of poset \cite{cicalese2014improved} and decision tree \cite{li2020tight}.

\begin{theorem}\label{the:IGS_tree_approximation_ratio}
    If the input hierarchy is a tree, the greedy policy provides an approximation ratio of $\frac{1+\sqrt{5}}{2}$ for AIGS.
\end{theorem}

\begin{proof}
    Let $\widetilde{D}$ be the decision tree constructed by greedy policy and $D^\ast$ be the optimal decision tree, with respect to the cost function defined in \eqref{eq:cost1}. For such a tree-like poset, \citet{cicalese2014improved} proved that the cost of $\widetilde{D}$ satisfies
    \begin{equation}
    cost(\widetilde{D}) \le \frac{1+\sqrt{5}}{2}\cdot cost(D^\ast).
    \end{equation}
    Furthermore, according to Lemma~\ref{the:IGS_to_searchinposet_appendix}, we know that $cost(\widetilde{D})$ is exactly the expected number of questions asked by the greedy policy for AIGS. Putting it together completes the proof.
\end{proof}

\begin{theorem}\label{the:IGS_equal_prob_approximation_ratio}
	If every node has an equal probability to be the target node, i.e., $p(v)=\frac{1}{n}$ for every $v\in V$, the greedy policy provides an approximation ratio of $O\left(\frac{\log n}{\log \log n}\right)$ for AIGS.
\end{theorem}

\begin{proof}
    \citet{li2020tight} showed that for any instance of the decision tree problem on $n$ objects with equal probability to be the target, the greedy decision tree $\widetilde{D}$ and the optimal decision tree $D^\ast$ satisfy
    \begin{equation*}
    cost(\widetilde{D}) \le \frac{6 \log n}{\log cost(D^\ast)}\cdot cost(D^\ast).
    \end{equation*}
    Moreover, for any binary tree with $n$ leaf nodes, it is trivial to verify that the total depth of all leaf nodes is minimized on the complete binary tree, which is at least $n\log_2 n$. This implies the optimal solution has a cost of at least $\log_2 n$. As a result, the greedy policy returns an $O\left(\frac{\log n}{\log \log n}\right)$-approximate solution for AIGS when every node has an equal probability to be the target node. Hence, the theorem is proved.
\end{proof}

\subsection{Extension to Heterogeneous Cost}\label{section:Analysis:CAIGS}
In the above discussion, we consider that each query charges a fixed price. In some scenarios, different questions may have different difficulties, which asks for heterogeneous payments for different questions, e.g.,~$\$0.5$ for an easy question and $\$1.5$ for a hard question. We show that the greedy policy with a slight modification on the definition of \textit{middle point}, taking the cost weight of each query into consideration, can address the problem with strong theoretical guarantees.

Consider the cost-sensitive AIGS (CAIGS) problem that the query on node $v$ charges a price of $c(v)$. It is easy to observe that the cost for identifying a target node $z$ is equal to the total cost of the query nodes on the path from the root to the leaf node $z$ in $D$. Similar to Definition~\ref{def:cost1}, the cost of a query policy for CAIGS can be defined as follows.
\begin{definition}[Cost of Query Policy for CAIGS]\label{def:cost-sensitive-cost}
	Given a query policy for the CAIGS problem, let $D$ be the decision tree constructed according to the policy. Denote by $L(D)$ the set of all leaf nodes in $D$, and by $\hat{\ell}(v)$ the total cost of the query nodes on the path from the root to the leaf node $v$ in $D$. Then, the cost of $D$, i.e.,~the expected cost of the query policy, is given by
	\begin{equation}\label{eq:cost-sensitive-cost}
		cost(D)=\sum_{v \in L(D)}p(v)\hat{\ell}(v).
	\end{equation}
\end{definition}


To cope with the heterogeneous cost of each query node, we attempt to select the query node such that (i) it splits the category graph as evenly as possible with respect to the probability weights, and meanwhile (ii) its query cost is as cheap as possible. For goal (i), we try to find a node $u$ that maximizes $p(G_u)p(G \setminus G_u)$. That is, since $p(G_u)+p(G \setminus G_u)=p(G)$ is a constant, maximizing $p(G_u)p(G \setminus G_u)$ is to minimize the difference between $p(G_u)$ and $p(G \setminus G_u)$, i.e., $\min \abs{p(G_u)-p(G \setminus G_u)}$. For goal (ii), we directly minimize $c(u)$. Finally, to optimize both goals simultaneously, we select the node $u$ that maximizes $\frac{{p(G_u)p(G \setminus G_u)}}{c(u)}$. We define such a node as the cost-sensitive middle point.

\begin{definition}[Cost-Sensitive Middle Point]\label{def:cost_sensitive_middle_point}
	Given a DAG $G$ associated with a probability weight $p(v)$ for $v$ being the target node and a query cost $c(v)$ on node $v$, the cost-sensitive middle point $u^\ast$ of $G$ is defined as
	\begin{equation*}
		u^\ast:=\argmax\limits_{u \in V}\ \frac{{p(G_u)p(G \setminus G_u)}}{c(u)}.
	\end{equation*}
\end{definition}

\begin{figure}[!bpt]
	\centering
	\subfloat[Hierarchy]{\includegraphics[height=0.115\textheight]{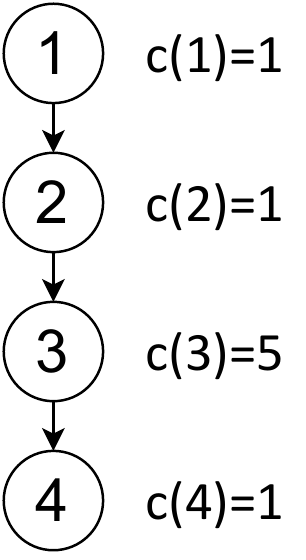}\label{fig:CAIGS-Example}}\hfil
	\subfloat[Simple Greedy]{\includegraphics[height=0.115\textheight]{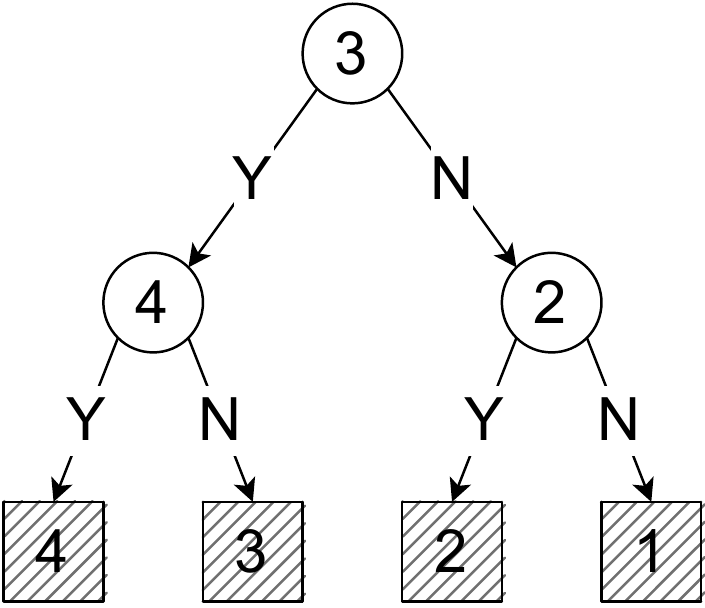}\label{fig:CAIGS-DT1}}\hfil
	\subfloat[Cost-sensitive Greedy]{\includegraphics[height=0.115\textheight]{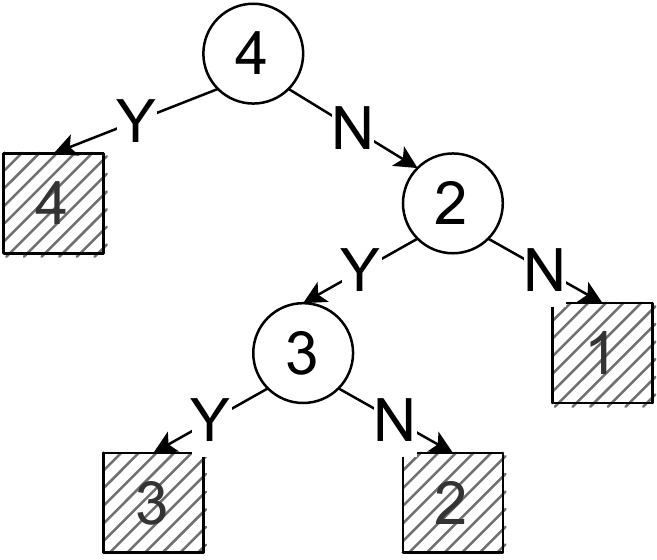}\label{fig:CAIGS-DT2}}
	\caption{Cost-sensitive greedy for CAIGS.}
	\label{fig:CAIGS}
\end{figure}

Note that when each query charges a unit price, the node $u^\ast$ maximizing $p(G_u)p(G \setminus G_u)$ also minimizing $\abs{p(G_u)-p(G \setminus G_u)}$, which generalizes the middle point defined in Definition~\ref{def:middle_point} for the case of nodes with homogeneous costs. In the following, we show that the cost-sensitive rounded greedy algorithm, which picks the cost-sensitive middle point as the query node in each round with respect to the rounded probability weights by equation (\ref{equation:rounded-greedy}), achieves an approximation ratio of $2(1+3\ln n)$ for CAIGS.

\begin{theorem}\label{thm:cost-sensitive-round-greedy}
The cost-sensitive rounded greedy algorithm provides an approximation ratio of $2(1+3\ln n)$ for CAIGS.
\end{theorem}

\begin{proof}[Proof (Sketch)] \citet{adler2008approximating} showed that when the probability weight $p(v)$ is identical for every $v$, i.e., $p(v)=1/n$, the cost-sensitive greedy policy considering heterogeneous cost for each test can obtain the same approximation ratio as homogeneous cost on the binary decision tree problem. Combing with the rounding technique~\cite{chakaravarthy2007decision} tailored to the weighted probability $p(v)$ yields that the cost-sensitive rounded greedy algorithm has an approximation ratio of $2(1+3\ln n)$ for the CAIGS problem.
\end{proof}

The following example demonstrates the effectiveness of the cost-sensitive rounded greedy algorithm.

\begin{example}
As an example, consider a simple category hierarchy given in \figurename~\ref{fig:CAIGS-Example}. Without considering the query cost of each node (i.e., homogeneous costs), as shown in \figurename~\ref{fig:CAIGS-DT1}, the greedy strategy first selects node 3; then chooses node 4 if a yes-answer is returned to further distinguish nodes 3 and 4, and chooses node 2 otherwise to distinguish nodes 1 and 2. Suppose that $c(1)=c(2)=c(4)=1$ and $c(3)=5$. Then, the expected cost of such a greedy strategy is $5+1\times 0.5+1\times 0.5=6$. On the other hand, as shown in \figurename~\ref{fig:CAIGS-DT2}, the cost-sensitive greedy algorithm will first select node 4 rather than node 3, since $\frac{p(G_4) p(G \setminus G_4)}{c(4)}=\frac{0.25\times 0.75}{1}=0.1875$ while $\frac{p(G_3) p(G \setminus G_3)}{c(3)}=\frac{0.5\times 0.5}{5}=0.05$. If a no-answer is returned (no further action is required if a yes-answer is returned), it selects node 2. Only if a yes-answer is further returned, it chooses node 3. Thus, the expected cost of the cost-sensitive greedy algorithm is $1+1\times 0.75+5\times 0.5=4.25$, which is significantly smaller than $6$ by the simple greedy algorithm without considering the query cost of each node. \done
\end{example}

\begin{algorithm}[!tbp]
	\setlength{\hsize}{0.94\linewidth}
	\caption{$\greedyN(G)$}
	\label{alg:BIGS_naive}
	\KwIn{an input DAG $G$}
	\KwOut{the target node}
	$sum\_prob \gets 1$\;
	\While{$\abs{G}>1$}{
		$min\_prob \gets +\infty$\;\label{algline:q_start}
		\ForEach{node $v \in G$}{
			$reach\_prob \gets \getp(G,v)$\;\label{algline:q_prob}
			\If{$\abs{2\cdot reach\_prob-{sum\_prob}} < min\_prob$}{
				$q \gets v$\;
				$q\_prob \gets reach\_prob$\;
				$min\_prob\gets \abs{2\cdot reach\_prob-{sum\_prob}}$\;\label{algline:q_end}
			}
		}
		\uIf{$reach(q)=yes$}{
			$G\gets G_q$\;\label{algline:update-Gq}
			$sum\_prob \gets q\_prob$\;
		}
		\Else{
			$G\gets G\setminus G_q$\;\label{algline:update-not-Gq}
			$sum\_prob \gets sum\_prob-q\_prob$\;
		}
	}
	\Return{the only node in $G$\;}
\end{algorithm}

\subsection{Additional Discussions}\label{subsec:discussions}
In the above analysis, we prove the approximation guarantees in the most general case. However, there may be special cases where the theoretical and experimental results are significantly better. For example, \citet{nowak2008generalized} showed that the greedy algorithm can achieve a better guarantee given that the hierarchy has some geometric conditions. However, such conditions are hard to meet in practice. Deriving a succinct condition is certainly an interesting topic.

In this paper, we only consider a purely sequential query policy. Apparently, queries may be asked in a batch to reduce interactions. Designing effective batched algorithms is an important extension and is rather challenging. For AIGS on a tree, we can ask a batch of $k$ questions simultaneously leveraging the $k$-partition scheme~\cite{kundu1977linear} to ensure provable guarantees. However, for AIGS on a general DAG, it remains an open problem for devising effective algorithms in a batched version with bounded guarantees.

Furthermore, our algorithms are designed for regular crowds. Some companies may employ in-house experts in practice. In such an expert environment, we can ask some more complicated queries which are difficult to the regular crowd but tractable to their in-house experts. There is a trade-off between the difficulty and the number of queries. Constructing a flexible strategy that properly balances the trade-off is an interesting problem, and we leave it as the future work. 

Note also that the AIGS problem is motivated by an application of image categorization, but our algorithms can be also applied to other content formats, e.g., textual tasks. Some existing experiment results~\cite{jain2017understanding} revealed that textual understanding often takes more time than visual. Therefore, the cost is likely to be higher if directly applying our approach to textual tasks.

%% file: chapters/05-Algorithm.tex
\section{Instantiations of Greedy Policy}\label{sec:algorithm}
A key challenge for instantiating the greedy policy is to find the \textit{middle point} of the candidate hierarchy $G$. In this section, we first present a naive instantiation by independently calculating the total weight of all reachable nodes for every candidate node $v\in G$. To make use of the graph structure of the candidate hierarchy, we then design efficient algorithms to instantiate the greedy policy.

\begin{algorithm}[!tbp]
	\caption{$\getp(G,v)$}
	\label{alg:get_subgraph_weight}
	\KwIn{a DAG $G$ and a node $v$}
	\KwOut{the total probability of all $v$'s reachable nodes in $G$}
	$prob \gets 0$\;
	initialize a queue $S$ and insert the node $v$ into $S$\;
	mark every node in $G$ as \textit{unvisited}\;
	\While{$S$ is not empty}{
		push out a node from $S$ as $u$\;
		\ForEach{child $v$ of $u$}{
			\If{$v$ is not visited}{
				$prob \gets prob+p(v)$\;
				insert $v$ into $S$\;
				mark $v$ as \textit{visited}\;
			}
		}
	}
	\Return{$prob$\;}
\end{algorithm}

\subsection{A Naive Instantiation}
Algorithm~\ref{alg:BIGS_naive} gives the pseudo-code for a naive implementation of the greedy policy, referred to as \greedyN. It simply enumerates all candidate nodes and computes the total probability of all nodes reachable from every node independently in each round of query (Line~\ref{algline:q_start}--\ref{algline:q_end}). In particular, the total probability of the subgraph rooted at $v$ is calculated via a subroutine \getp (Line~\ref{algline:q_prob}) given by Algorithm~\ref{alg:get_subgraph_weight}. Given a graph $G$ and a query node $v$, the \getp algorithm performs a breadth first search (BFS) that starts from $v$. However, this naive instantiation can introduce heavy computational overhead.

\spara{Time Complexity} During the search process, as each round eliminates at least one candidate node, there will be at most $n$ rounds of finding the middle point. In each round, it will enumerate at most $n$ nodes in the candidate set and take $O(m)$ time for computing the total probability of the subgraph rooted at each candidate node, where $m$ is the number of edges in $G$. After finding the middle point, the algorithm needs to update the graph by performing a BFS according to the answer with $O(m)$ time. Hence, the total time complexity of \greedyN is $O(n^2m)$.

\spara{Discussion} The aforementioned instantiation of the greedy policy is straightforward and intuitive, but it is far from optimized in terms of its efficiency. In fact, such a naive implementation may not handle large-scale hierarchies due to its relatively high time complexity. To tackle the efficiency issue, there are two key challenges to be solved\textemdash (i) how to find the middle point efficiently and (ii) how to update the graph efficiently after getting the query result.

\SetKw{And}{and}
\begin{algorithm}[!tbp]
	\caption{$\greedyT(T)$}
	\label{alg:BIGS_T}
	\KwIn{a tree $T$}
	\KwOut{the target node}
	$r\gets$ the root of $T$\;
	initialize $\tilde{p}(v)\gets p(T_v)$ and $size(v)\gets \abs{T_v}$ for each $v\in T$ by performing $\getpDFS(T,r)$ (Algorithm~\ref{alg:setweight})\;\label{algline:precompute}
	\While{$size(r)>1$}{
		$v \gets r$\;\label{alg:BIGS_T:find_middle_1}
		\While{$2\tilde{p}(v)>\tilde{p}(r)$ \And $v$ is not a leaf node}{
			$u\gets v$\;
			$v\gets$ the child of $u$ with the largest $\tilde{p}(v)$\; 
			\label{alg:BIGS_T:find_hc}
		}
		\lIf{$\abs{2\tilde{p}(u)-\tilde{p}(r)} \le \abs{2\tilde{p}(v)-\tilde{p}(r)}$}{$q \gets u$}
		\lElse{$q \gets v$\label{alg:BIGS_T:find_middle_2}}
		\lIf{$reach(q)=yes$\label{alg:BIGS_T:update}}{$r \gets q$\label{alg:BIGS_T:change_root}}
		\Else{
			\ForEach{node $v $ on the path from $r$ to $q$\label{alg:BIGS_T:update_weights}}{
				$\tilde{p}(v) \gets \tilde{p}(v) - \tilde{p}(q)$\;
				$size(v) \gets size(v) - size(q)$\;\label{alg:BIGS_T:update_weights_end}
			}
		}
	}
	\Return{node $r$}\;
\end{algorithm}

\begin{algorithm}[!tbp]
	\caption{$\getpDFS(T,u)$}\label{alg:setweight}
	\KwIn{a tree $T$ and a node $u$}
	\KwOut{set the probability $\tilde{p}(u)=p(T_u)$ of subtree $T_u$ rooted at $u$ and its size $size(u)=\abs{T_u}$}
	initialize $\tilde{p}(u)\gets p(u)$ and $size(u)\gets 1$\;
	\If{node $u$ is not a leaf}{
		\ForEach{child $v $ of $u$}{
			$\getpDFS(T,v)$\;
			update $\tilde{p}(u)\gets \tilde{p}(u)+\tilde{p}(v)$\;
			update $size(u)\gets size(u)+size(v)$\;
		}
	}
\end{algorithm}

\subsection{Efficient Instantiation on Tree}
We first look at a simple scenario when the input hierarchy is a tree and denote such a hierarchy as $T$ (instead of $G$). In the \greedyN algorithm, it may be quite time-consuming by performing BFS to find the middle point. We address this challenge utilizing the notion of \textit{weighted heavy path} which extends the concept of heavy-path~\cite{sleator1983data} taking into account the probability of each node. In the following, we formally define the weighted heavy path.

\begin{definition}[Weighted Heavy Path]
Given an internal node $u$ of $T$, let $v$ be a child of $u$ with the largest subtree weight\footnote{The weight of a subtree is the total probability of nodes therein.} (with ties broken arbitrarily). Then, the edge between $u$ and $v$ is said to be \textit{heavy}, and the other out-going edges of $u$ are said to be \textit{light}. A weighted heavy path is a maximal path by concatenating heavy edges, i.e., the path cannot be extended with another heavy edge.
\end{definition}

According to the above definition, it is easy to see that for each node $u$ in $T$, there is at most one in-coming/out-going edge of $u$ is heavy. This implies that every node appears in one and exactly one weighted heavy path. We denote by $H(T,u)$ the weighted heavy path of $T$ that contains the node $u$. We show that $H(T,r)$ contains a middle point of $T$, where $r$ is the root of $T$.

\begin{theorem}\label{thm:Middle_Point_on_Heavy_Path}
Given a tree $T$, let $H(T,r)$ be the weighted heavy path containing the root $r$ of $T$, with respect to the node probability. Then, the node $u\in H(T,r)$ that minimizes $\abs{2\cdot p(T_u)-p(T)}$ is the middle point of $T$, where $T_u$ denotes the subtree of $T$ rooted at $u$.
\end{theorem}
\begin{proof}
	Consider any internal node $u\in T$, let $v$ be the child of $u$ such that the edge $(u,v)$ is heavy and $x$ be any other child of $u$. Clearly, it holds that $p(T_u)\geq 2p(T_x)$, since $p(T_v)\geq p(T_x)$ and $p(T_u)\geq p(T_v)+p(T_x)+p(u)$. Thus,
	\begin{equation*}
		\abs{2p(T_x)-p(T)}=p(T)-2p(T_x).
	\end{equation*}
	Moreover, for any descendant $y$ of $x$, we have $p(T_x)\geq p(T_y)$, which indicates that
	\begin{equation*}
		\abs{2p(T_y)-p(T)}=p(T)-2p(T_y)\geq p(T)-2p(T_x).
	\end{equation*}
	
	Now, we consider the node $v$. There are two cases: (i) $p(T)\geq 2p(T_v)$ and (ii) $p(T)< 2p(T_v)$. For case (i), obviously,
	\begin{equation*}
		\abs{2p(T_v)-p(T)}=p(T)-2p(T_v)\leq p(T)-2p(T_x).
	\end{equation*}
	For case (ii), as $p(T_v)+p(T_x)\leq p(T)$, it also holds that
	\begin{equation*}
	\abs{2p(T_v)-p(T)}=2p(T_v)-p(T)\leq p(T)-2p(T_x).
	\end{equation*}
	
	Putting it together, we have
	\begin{equation*}
		\abs{2p(T_v)-p(T)}\leq \abs{2p(T_x)-p(T)}\leq \abs{2p(T_y)-p(T)}.
	\end{equation*}
	This implies that $v$ dominates all other children of $u$ as well as all their descendants. Hence, starting with $u=r$, we just keep the subtree rooted at the child $v$ of $u$ such that the edge $(u,v)$ is heavy, and meanwhile remove all subtrees rooted at the other children of $u$, since all nodes in latter are dominated by $v$. We then set $u=v$ and repeat the process until $v$ is a leaf node. This procedure actually generates the weighted heavy path $H(T,r)$ containing the root $r$ of $T$. Therefore, by definition, the node $u\in H(T,r)$ that minimizes $\abs{2\cdot p(T_u)-p(T)}$ is the middle point of $T$.
\end{proof}

\begin{algorithm}[!tbp]
	\setlength{\hsize}{0.95\linewidth}
	\caption{$\greedyG(G)$}
	\label{alg:bigs_g}
	\KwIn{an input DAG $G$}
	\KwOut{the target node}
	round $w(v) \gets \left\lceil\frac{n^2\cdot p(v)}{\max_{v \in V} p(v)}\right\rceil$ for each $v$\;\label{algline:round-end}
	initialize $r \gets$ root of $G$, and $\tilde{w}(v)\gets w(G_v)$ for each $v$\;\label{algline:weight-initial-end}
	\While{$r$ is not a leaf}{
		initialize $q \gets r$, queue $S \gets \{r\}$, and $min\_w \gets +\infty$\;\label{alg:bigs_g:bfs_1}
		\While{$S$ is not empty}{
			pop out a node from $S$ as $u$\;
			\ForEach{child $v$ of $u$}{
				\If{$\abs{2\tilde{w}(v)-\tilde{w}(r)} < min\_w$}{
					$q \gets v$\;
					$min\_w\gets \abs{2\tilde{w}(v)-\tilde{w}(r)} $\;
				}
				\lIf{$2\tilde{w}(v)>\tilde{w}(r)$\label{alg:bigs_g:early_stop}}{insert $v$ into $S$\label{alg:bigs_g:bfs_2}}
			}
		}
		\lIf{$query(q)=yes$}{$r \gets q$\label{alg:bigs_g:adjust_root}}
		\Else{
			\lForEach{$v \in G_q$}{$\adjust(G,v)$\label{alg:bigs_g:adjust_weights}}
			delete every node in $G_q$ from $G$\;\label{alg:bigs_g:delete_node}
		}
	}
	\Return{node $r$\;}
\end{algorithm}

Based on the Theorem~\ref{thm:Middle_Point_on_Heavy_Path}, we can find the middle point by considering the nodes on the weighted heavy path to avoid enumerating all nodes in the whole tree. Algorithm~\ref{alg:BIGS_T} presents our improved algorithm for the instantiation of the greedy policy for AIGS on trees, namely \greedyT. Initially, it invokes \getpDFS (Algorithm~\ref{alg:setweight}) by performing one round depth-first search (DFS) starting from the root $r$ of $T$ to compute $\tilde{p}(v)=p(T_v)$ and $size(v)=\abs{T_v}$ for every node $v \in T$ (Line~\ref{algline:precompute}). To find the middle point, \greedyT traverses the weighted heavy path using a top-down strategy starting from the root of tree (Lines~\ref{alg:BIGS_T:find_middle_1}--\ref{alg:BIGS_T:find_middle_2}). To update the candidate tree after we get the answer, instead of being fully recomputed, these weights are incrementally updated only if the answer to the query is \textit{no} (Lines~\ref{alg:BIGS_T:update}--\ref{alg:BIGS_T:update_weights_end}). That is, for every node $v$ on the path from the root $r$ to the query node $q$, its weight $\tilde{p}(v)$ will be decreased by a value of $\tilde{p}(q)$ after removing subtree $T_q$, while the weight $\tilde{p}(u)$ remains the same for any $u$ not on this path. On the other hand, if the answer is \textit{yes}, the algorithm just simply sets the query node $q$ as the new root $r$.

\spara{Time Complexity} Computing the weight of each subtree via \getpDFS (Algorithm~\ref{alg:setweight}) takes $O(n)$ time by performing one round DFS. Denote by $d$ the maximum out-degree of the nodes and by $h$ the length of the longest path in $T$. In each round of query, \greedyT starts from the root and goes down along the weighted heavy path for at most $h$ layers and the node in each layer has at most $d$ children. Thus, the time complexity for the query node selection in each round is $O(hd)$. After obtaining the query response, it takes $O(h)$ time to update the weights for all ancestors of the query node. Therefore, the total time complexity of \greedyT is $O(nhd)$.\footnote{If a max-heap is used to store the weights of the children of each node, one can verify that \greedyT just takes $O(nh\log d)$ time.}

\subsection{Efficient Instantiation on DAG}
We then move to the general case when the hierarchy is a DAG. Algorithm~\ref{alg:bigs_g} instantiates the rounded greedy policy with $2(1+3\ln n)$-approximation (see Theorem~\ref{thm:round-greedy}), namely \greedyG. It first initializes $\tilde{w}(v)=w(G_v)$ by \getp (Algorithm~\ref{alg:get_subgraph_weight}) with respect to the rounded weight $w(\cdot)$ (Lines~\ref{algline:round-end}--\ref{algline:weight-initial-end}). Then, it goes down from the root to search for the middle point by BFS (Lines~\ref{alg:bigs_g:bfs_1}--\ref{alg:bigs_g:bfs_2}). Specifically, if the total weight $\tilde{w}(v)$ of reachable nodes from $v$ satisfies $2\tilde{w}(v)\leq \tilde{w}(r)$, where $r$ is the root of current candidate DAG, then $v$ dominates all its descendants. As a result, all descendants of such a node $v$ are skipped. Upon receiving the query answer, the graph updates the weights of $v$'s ancestors for each $v\in G_q$ only if the answer is \textit{no} (Lines~\ref{alg:bigs_g:adjust_root}--\ref{alg:bigs_g:delete_node}). In particular, it calls \adjust (Algorithm~\ref{alg:bigs_g_reverse_deletion}) to update the weights by performing a reverse BFS from each $v\in G_q$.

\begin{algorithm}[!tbp]
	\caption{$\adjust(G,v)$}
	\label{alg:bigs_g_reverse_deletion}
	\KwIn{a DAG $G$ and the node $v$ to be deleted}
	initialize a queue $S$ and insert the node $v$ into $S$\;
	mark every node in $G$ as \textit{unvisited}\;
	\While{$S$ is not empty}{
		pop out a node from $S$ as $u$\;
		\ForEach{each parent $v$ of $u$}{
			\If{$v$ is not visited}{
				$\tilde{w}(v) \gets \tilde{w}(v)-w(v)$\;
				insert $v$ into $S$ and mark $v$ as \textit{visited}\;
			}
		}
	}
\end{algorithm}

\spara{Time Complexity} The total time complexity for computing $\tilde{w}(v)$ for every node is $O(nm)$. For each query, it takes $O(m)$ time to identify the middle point. Since there will be at most $n$ rounds of finding the middle point, the total time used for identifying query nodes is $O(nm)$. Finally, to remove a node from the graph after receiving the query result, it takes $O(m)$ time to perform a reverse BFS to update the graph via \adjust (Algorithm~\ref{alg:bigs_g_reverse_deletion}). Since it at most removes $n$ nodes, the total update time is $O(nm)$. Therefore, the total time complexity of \greedyG is $O(nm)$, which significantly improves the \greedyN algorithm of $O(n^2m)$.

%% file: chapters/06-Experiment.tex
\section{Experiments}\label{sec:experiment}
In this section, we evaluate the performance of our proposed approach with extensive experiments on two real-world datasets. 
The key metrics used for comparison are cost (i.e., the number of queries) and running time.
All experiments are conducted on a machine with an Intel i7-7700 CPU and 32GB RAM. All the algorithms are implemented in Python.

\subsection{Experimental Setting}\label{sec:exp:setting}

\spara{Datasets} Similar to previous work~\cite{li2020efficient,Tao:IGS}, we test our algorithms using the following two real-world datasets, including Amazon and ImageNet. Table~\ref{tab:dataset_attributes} summarizes some important attributes of the tested datasets.
\begin{itemize}[topsep=2mm, partopsep=0pt, itemsep=1mm, leftmargin=18pt]
	\item Amazon~\cite{he2016ups}. This is a  dataset including the product hierarchy at Amazon with a tree structure. This dataset contains information of products sold at Amazon, like reviews and product metadata. Specifically, the record has a field named \textit{categories}, and we can consider this field as a path starting from the root of the hierarchy to this product category. By combining these paths together, we can get a tree hierarchy with 29,240 nodes.
	
	\item ImageNet~\cite{deng2009imagenet}. This is a large-scale hierarchical image dataset using the structure of WordNet~\cite{miller1998wordnet}. Each category is represented as a tag called as \textit{synset} in the XML file (\url{https://www.imagenet.org/api/xml/structure_released.xml}), and the subcategory is given inside each tag explicitly. The attribute \textit{wnid} of each tag assigns an unique id to each category. We note that the category with $\textit{wnid}=\text{``fa11misc''}$ contains miscellaneous images that do not conform to WordNet. We extract all categories except the one with $\textit{wnid}=\text{``fa11misc''}$ and get a DAG with 27,714 nodes.
\end{itemize}

\spara{Metrics} The primary metric evaluated in this section is the \textit{cost}, i.e.,~the expected number of queries. Moreover, to evaluate the efficiency of proposed instantiations of the greedy policy, a comparison of running time will be included as well.

\begin{table}[!tbp]
	\caption{Statistics of datasets.}
	\label{tab:dataset_attributes}
	\centering
	\setlength{\tabcolsep}{0.6em} 
	\renewcommand{\arraystretch}{1.1}
	\begin{tabular}{lrcccc}
		\toprule
		\textbf{Dataset} & \textbf{\#nodes} & \textbf{Height} & \textbf{Max Deg.} & \textbf{Type} & \textbf{\#objects}\\
		\midrule
		Amazon & 29,240 & 10 & 225 & Tree & 13{,}886{,}889 \\
		ImageNet & 27,714 & 13 & 402 & DAG & 12{,}656{,}970\\
		\bottomrule
	\end{tabular}
\end{table}

\spara{Competing Algorithms} 
We compare our \greedyT (on Amazon) and \greedyG (on ImageNet) algorithms against three baselines, including the native \topdown method, the heavy-path-based binary search method for the WIGS problem proposed by \citet{Tao:IGS}, referred to as \igs, and the search method using multiple-choice query proposed by \citet{li2020efficient}, referred to as \migs. For \migs~\cite{li2020efficient}, we consider the number of choices read by the crowd as the cost, since a $k$-choice query can be decomposed to $k$ binary queries.

\subsection{Experimental Results}\label{sec:exp:query_complexity}
\begin{table}[!tbp]
	\centering
    \caption{Cost under real data distribution.}
    \label{tab:TotalQueries_RealDistribution}
	\setlength{\tabcolsep}{0.6em} 
	\renewcommand{\arraystretch}{1.1}
	\begin{tabular}{lcccc}
		\toprule
	 	\textbf{Dataset} & {\topdown} & \MIGS & \igs & \textbf{\greedyT/\greedyG}\\
		\midrule
        Amazon & 92.23 & 89.19 & 37.35 & \textbf{21.02}\\
        ImageNet & 101.18 & 96.28 & 30.18 & \textbf{22.29}\\
		\bottomrule
	\end{tabular}
\end{table}

\subsubsection{Comparison of Cost}
We first compare the cost of different algorithms. In the experiments, we count the number of objects in each category to obtain the real probability distribution. However, the real probability distribution may not be known in practice. To tackle this issue, we apply a simple online learning approach that dynamically adjusts the empirical probability distribution on the fly upon obtaining the category result of each object. Furthermore, to demonstrate the robustness of our algorithms, we evaluate several synthetic probability distributions, including both homogeneous and heterogeneous probability settings.

\spara{Real Data Distribution}
Table~\ref{tab:TotalQueries_RealDistribution} shows the average number of queries asked by different algorithms for locating the target node given the a-priori known real data distribution. The result shows that our \greedyT and \greedyG methods significantly outperform all the three baselines. Specifically, compared with the naive \topdown method, our \greedyT and \greedyG algorithms save $77.21\%$ and $77.97\%$ cost on the Amazon and ImageNet datasets, respectively, while compared with \migs, our \greedyT and \greedyG algorithms save $76.43\%$ and $76.85\%$ cost on the two datasets. Meanwhile, we observe that, although \igs outperforms both \topdown and \migs, our \greedyT and \greedyG algorithms still achieve remarkable improvements of $43.72\%$ and $26.14\%$ on the two datasets. These results demonstrate the superiority of our proposed algorithms. 

Interestingly, our experiment results show that \topdown and \MIGS incur comparable cost. This is because the goal of \citet{li2020tight} is to minimize the expected number of multiple-choice queries. However, the workload to answer a multiple-choice query may vary significantly. For example, if \MIGS queries on the root node of ImageNet, a question with around $100$ choices is asked, while some other queries may only have a few choices. Hence, if we take into account the number of choices read by the crowd, we will find \MIGS can only reduce the crowd workload slightly by around $5\%$ compared with \topdown.

\spara{Learning Distribution on the Fly}
Instead of assuming the a-priori known real data distribution, we learn the empirical distribution on the fly. That is, when we label the $i$-th object, we use the statistics of the first $(i-1)$ labeled objects as the input probability distribution. At the very beginning when no object is labeled, we assume that all categories occur with an equal probability. We record the average cost for every 10 thousand objects. Note that when the objects arrive in different sequences, our algorithms will incur different costs. We generate 20 traces by randomly shuffling the objects and report the average results. For the sake of visualization, we include \igs, \greedyT and \greedyG using the offline distribution extracted from the real data as baselines for comparison, since the costs of \topdown, \migs are significantly higher. \figurename~\ref{fig:AveragePastCost_vs_progressness} shows the results. We observe that the average costs of the baselines, including \igs, \greedyT and \greedyG using the offline data distribution as an input, are almost the same as the values in Table~\ref{tab:TotalQueries_RealDistribution}. The reason is that the baseline algorithms do not rely on the online learned distribution and the distribution of every $10$ thousand objects roughly remains the same. As a contrary, the average cost of \greedyT and \greedyG using the online learned distribution decreases along with the number of objects labeled, and gradually converges to that using the real data distribution. Intuitively, the more categorization results we get, the more accurately we can predict the distribution of the future objects that will be labeled, which will eventually help reduce the cost. Interestingly, we find that although both datasets have more than 10 million objects, our algorithms equipped with a practical online learning approach perform very close (with a difference less than $3\%$) to those using the real data distribution just after obtaining the categories of 50 thousand objects.

\begin{figure}[!tbp]
	\centering
	\subfloat[Amazon]{\includegraphics[width=0.243\textwidth]{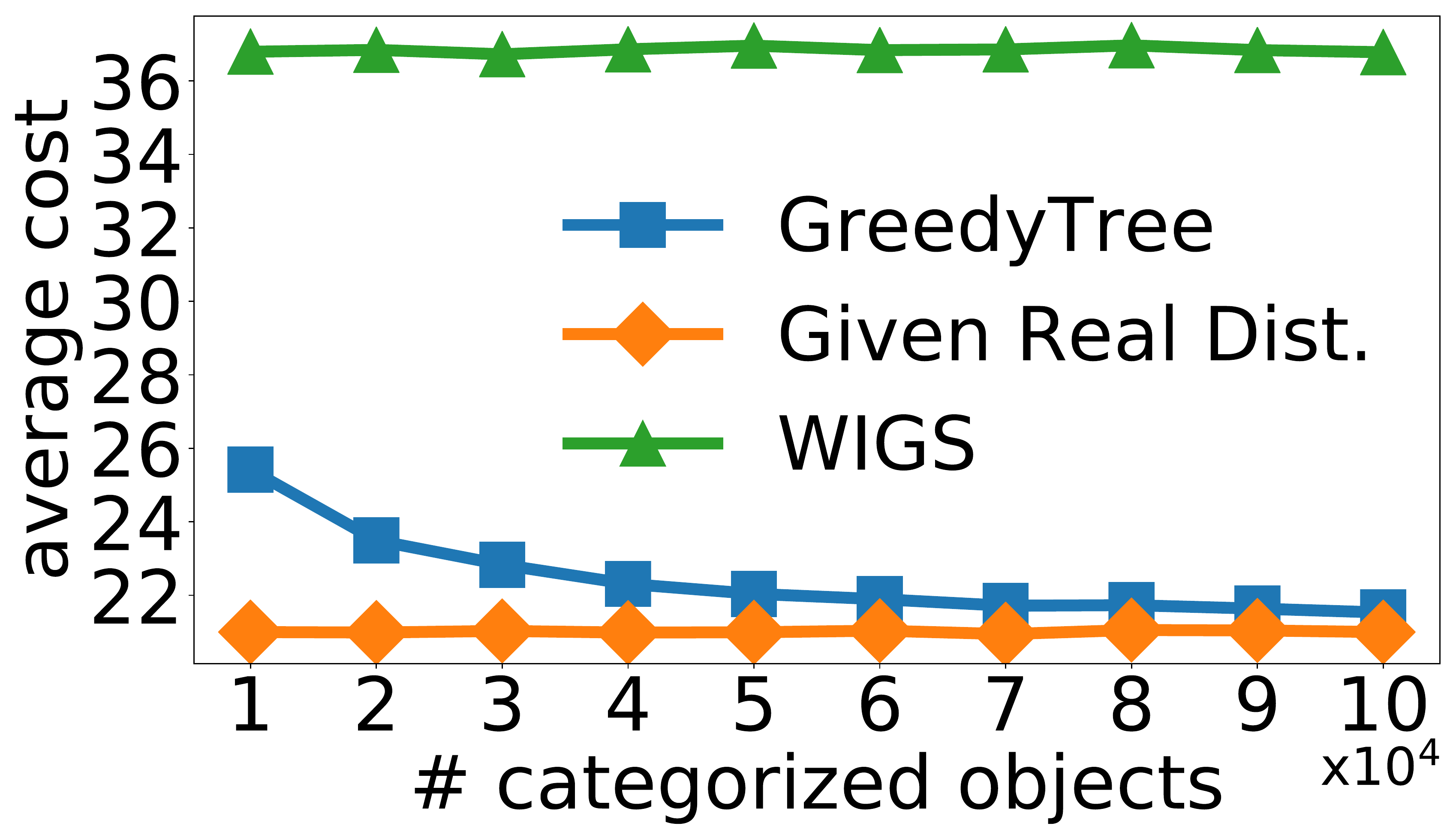}}\hfil
	\subfloat[ImageNet]{\includegraphics[width=0.243\textwidth]{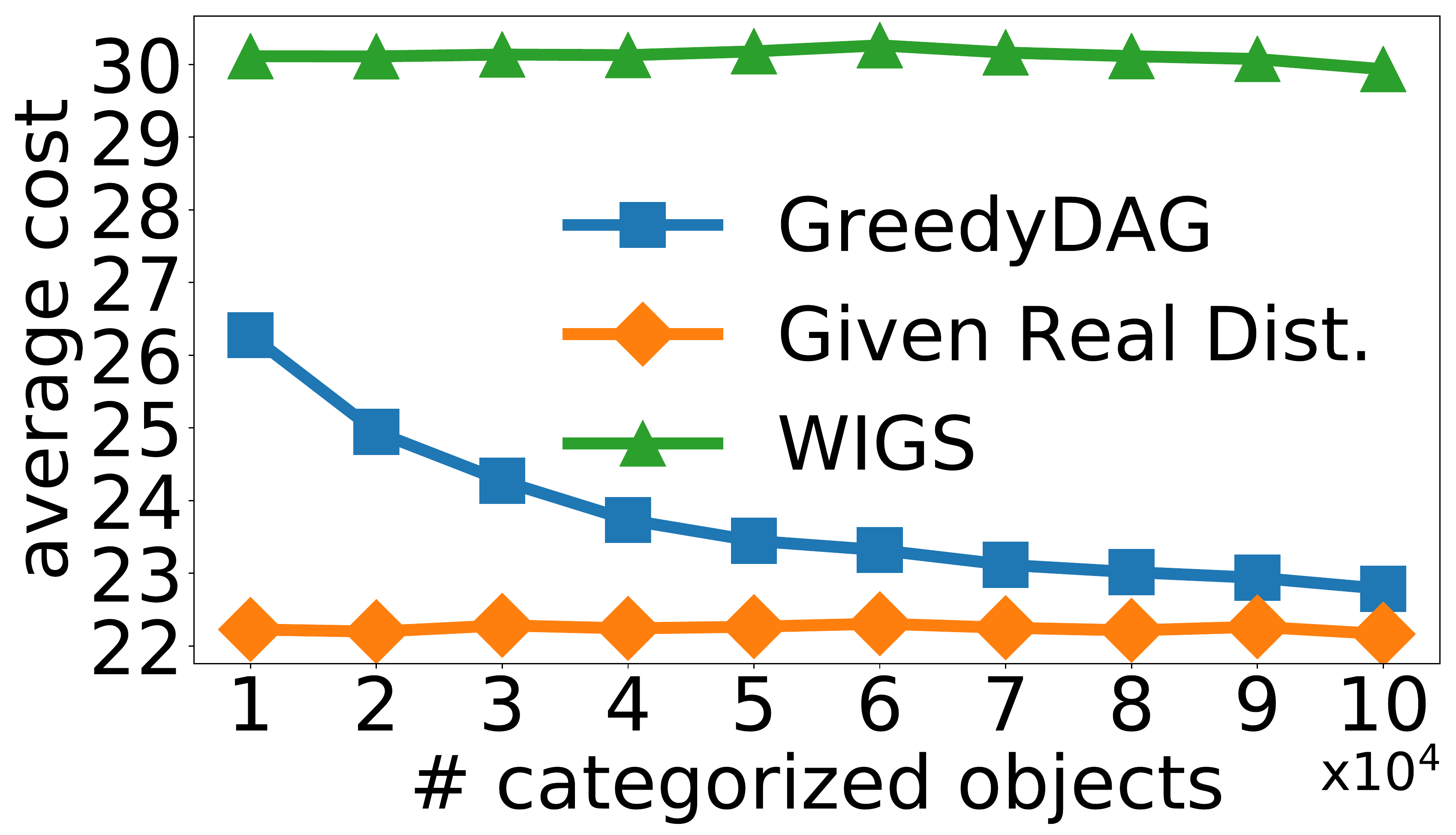}}
	\caption{Average cost vs.\ the number of categorized objects.}
	\label{fig:AveragePastCost_vs_progressness}
\end{figure}

\spara{Synthetic Data Distribution} To further evaluate the impact of data distribution, we use several representative synthetic probability distributions, including the unweighted setting and three weighted settings. For the unweighted setting, each node in the category hierarchy will be the target node with an equal probability, i.e., $p(v)=1/n$ for each $v$. For the weighted settings, we first randomly assign a value $x_v$ to each node $v$ according to certain distributions, including uniform distribution, exponential distribution, and Zipf distribution~\cite{zipf1932selected}, and then normalize $x_v$ as $v$'s occurring probability, i.e.,~$p(v)=\frac{x_v}{\sum_{v}x_v}$. Specifically, for uniform distribution, we first generate a random number $x_v$ uniformly distributed in the range of $(0,1)$ and then assign a node with the probability~$p(v)=\frac{x_v}{\sum_{v}x_v}$. Similarly, for exponential distribution, the random number is generated from the distribution $\operatorname{Exp}(1)$, while for Zipf distribution, the probability density function is $f(x; a)=\frac{x^{-a}}{\zeta(a)}$, where $\zeta$ is the Riemann Zeta function and $a$ is the parameter which is set as $a=2$ by default. Note that Zipf distribution is a long-tail probability distribution. Since each node is assigned with a random probability and the probability will influence the performance, we repeat the experiments 20 times and report the average result.

\begin{table}[!tbp]
    \centering
    \caption{Cost under several probability settings on Amazon.}
    \label{tab:TotalQueries_tree}
    \setlength{\tabcolsep}{0.6em} 
    \renewcommand{\arraystretch}{1.1}
	\begin{tabular}{lccccc}
	 	\toprule
	 	\textbf{Distribution} & {\topdown} & \MIGS & \igs & \textbf{\greedyT}\\
	 	\midrule
	 	Equal & 81.17 & 80.81 &27.42 & \textbf{25.35}\\
        Uniform & 81.28 & 81.19 & 27.47 & \textbf{23.68}\\
        Exponential & 82.42 & 81.65 & 27.37 & \textbf{22.70}\\
        Zipf & 82.09 & 81.94 & 27.55 & \textbf{14.03}\\
        \bottomrule
    \end{tabular}
\end{table}

\begin{table}[!tbp]
	\centering
    \caption{Cost under several probability settings on ImageNet.}
    \label{tab:TotalQueries_DAG}
	\setlength{\tabcolsep}{0.6em} 
	\renewcommand{\arraystretch}{1.1}
	\begin{tabular}{lcccc}
		\toprule
	 	\textbf{Distribution} & {\topdown} & \MIGS & \igs & \textbf{\greedyG}\\
		\midrule
		Equal & 123.31 & 126.12 &34.56 & \textbf{31.48}\\
        Uniform & 125.82 & 124.66 & 34.55 & \textbf{28.66}\\
        Exponential & 125.41 & 127.39 & 34.57 & \textbf{27.00}\\
        Zipf & 125.24 & 133.48 & 34.74 & \textbf{14.41}\\
		\bottomrule
	\end{tabular}
\end{table}

Table~\ref{tab:TotalQueries_tree} and Table~\ref{tab:TotalQueries_DAG} show the expected costs of different algorithms with various probability settings on the Amazon and ImageNet datasets, respectively. We observe that our \greedyT and \greedyG algorithms consistently and significantly outperform \topdown, \migs and \igs for all the four probability settings tested. Moreover, we find that the costs of \topdown and \igs almost remain the same for different distributions. The reason is that these baseline algorithms do not take into account the probability distribution and the expected probability associated with each node is the same as the generated probability is an independent and identically distributed random variable. On the contrary, the costs of our \greedyT and \greedyG methods are lower if the distribution of probability is more skewed. That is, the costs of \greedyT and \greedyG under the Zipf distribution are substantially less than those under the exponential distribution, and are in turns smaller than those under the uniform distribution, and are in turns further less than those under the unweighted setting. In particular, under the Zipf distribution where a few nodes have extremely large probabilities, the improvement percentage of \greedyT and \greedyG over \igs reaches $59.07\%$ and $58.52\%$ on the Amazon and ImageNet datasets, respectively, whereas for the unweighted setting, the reduction on cost is around $10\%$. These results demonstrate the superiority of our cost-effective design for the AIGS problem.

The experiment results in Table~\ref{tab:TotalQueries_tree} and Table~\ref{tab:TotalQueries_DAG} shows that a skewed probability distribution is in favor of our approach. We further confirm this observation by testing different parameters $a$ for the Zipf distribution. That is, the smaller the value parameter $a$ is, the more biased the distribution of probability is. \figurename~\ref{fig:weightedIGS_performance_zipf} gives the result. It clearly shows that the costs of \greedyT and \greedyG increase along with the parameter $a$ of the Zipf distribution, and finally approaches the costs when each node has an identical probability. This is because if a node is associated with a high probability, our approach will ask the question that can directly lead to this node, which can avoid many unnecessary queries.

\begin{figure}[!tbp]
	\centering
	\subfloat[Amazon]{\includegraphics[width=0.243\textwidth]{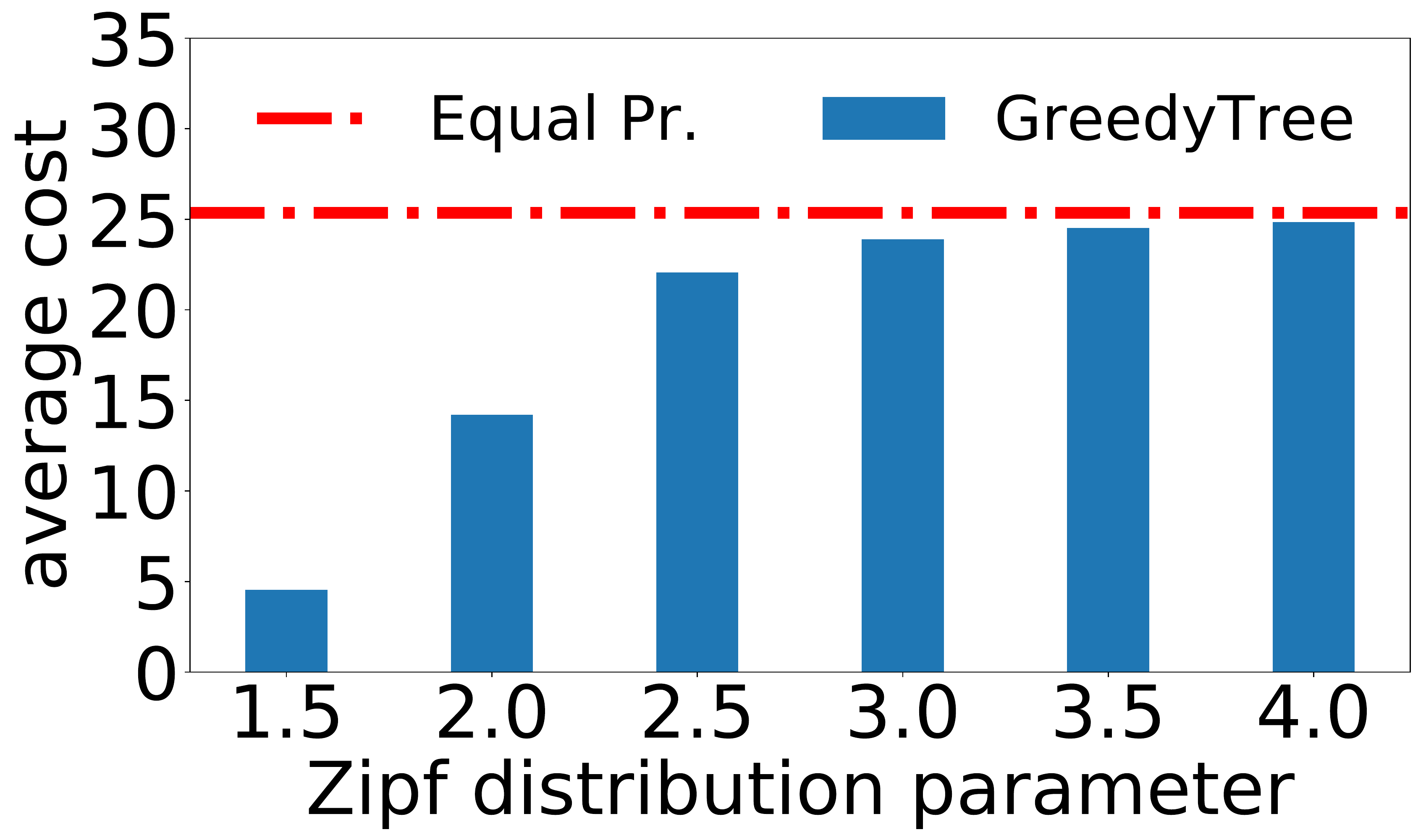}}\hfil
	\subfloat[ImageNet]{\includegraphics[width=0.243\textwidth]{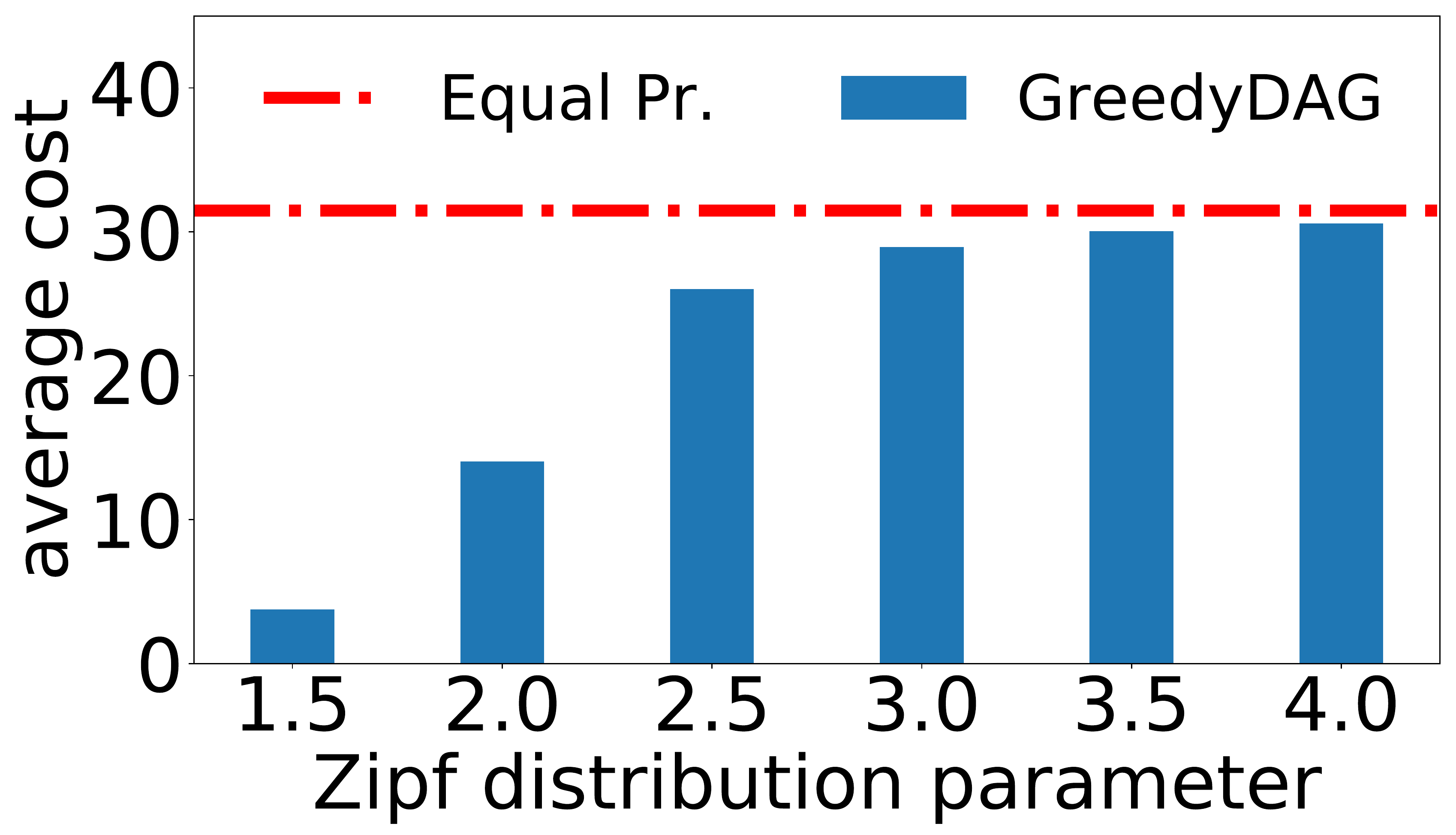}}
	\caption{Cost vs.\ parameter of Zipf distribution.}
	\label{fig:weightedIGS_performance_zipf}
\end{figure}

\subsubsection{Comparison of Running Time}\label{sec:exp:time_evaluation}
In this section, we evaluate the efficiency of our \greedyT and \greedyG algorithms and consider the \greedyN algorithm as a baseline for comparison. We randomly select $1{,}000$ nodes in each depth as the target nodes and compute the average running time. Note that one node may be chosen multiple times. For example, in depth $0$, there is only one node, i.e.,~the root, and the root will be selected $1{,}000$ times as the target node. \figurename~\ref{fig:time_uniformIGS_vs_layers} shows the average running time varying along with the node depth. On the Amazon dataset with a tree hierarchy, our \greedyT algorithm runs three orders of magnitudes faster than the \greedyN algorithm. On the ImageNet dataset with a general DAG hierarchy, our \greedyG algorithm is also noticeably faster than the \greedyN algorithm. These results demonstrate the efficiency of our proposed approach.

%% file: chapters/02-RelatedWork.tex
\section{Related work}
\subsection{Human-Based Computation}
Our work is related to human-based computation, which aims to solve tasks that are more challenging for computers but are easier to handle for humans. One idea is to ask for some information from humans and tackle the challenge with human intelligence. Crowdsourcing has emerged as a major technique that enables programmers to employ humans to solve these problems. There have been many works in the database area to study how to use human assistance to process data~\cite{marcus2011crowdsourced,franklin2011crowddb,liu2012cdas}, such as SQL-like query processing ~\cite{davidson2013using,park2012deco,karger2011human,li2017cdb}, retrieving the maximum item from a set~\cite{venetis2012max}, and finding the highest-ranked object~\cite{guo2012so}. Moreover, crowdsourcing is also applied in some other applications like data integration and data cleaning~\cite{wang2012crowder}, best path selection in geo-positioning services~\cite{zhang2014crowd}, and filtering data based on a set of properties~\cite{parameswaran2012crowdscreen}. Some similar tasks also occur in the machine learning area. In (supervised) machine learning, some algorithms require labeled data for training. However, the labeled data may be difficult or expensive to obtain. In active learning, the algorithm can choose the data from which it learns knowledge and requests the corresponding labels from human~\cite{dasgupta2004analysis,settles2009active}. This task is similar to our work in the sense that they all incorporate humans into the computation.

\subsection{Interactive Graph Search}
Interactive graph search (IGS) aims to find the target node on DAG with the assist of humans. The representative application of IGS is object categorization~\cite{chakrabarti2004automatic,li2020efficient}, hierarchy creation~\cite{chilton2013cascade,sun2015building}, and faceted search~\cite{basu2008minimum}. This problem is originally proposed by~\citet{Tao:IGS} and they devised algorithms using the heavy-path-based binary search technique with near-optimal theoretical bounds considering the \textit{worst-case} cost, referred to as WIGS which aims to minimize the cost in the worst-case. Prior to this work, \citet{parameswaran2011human} proposed the human-assisted graph search problem (HumanGS), which studied a graph search problem similar to IGS under the offline setting. Recently, \citet{li2020efficient} studied a variant IGS problem by asking multiple-choice questions and assuming the input hierarchy is a tree, namely MIGS. \citet{zhu2021budget} later studied the budget constrained interactive graph search for multiple targets, which considers that there are more than one target nodes on the hierarchy and the budget may not allow us to find the target nodes exactly. The techniques developed for these problems are ineffective for addressing our AIGS problem, due to the differences in problem definitions.

\begin{figure}[!tbp]
	\centering
	\subfloat[Amazon]{\includegraphics[width=0.243\textwidth]{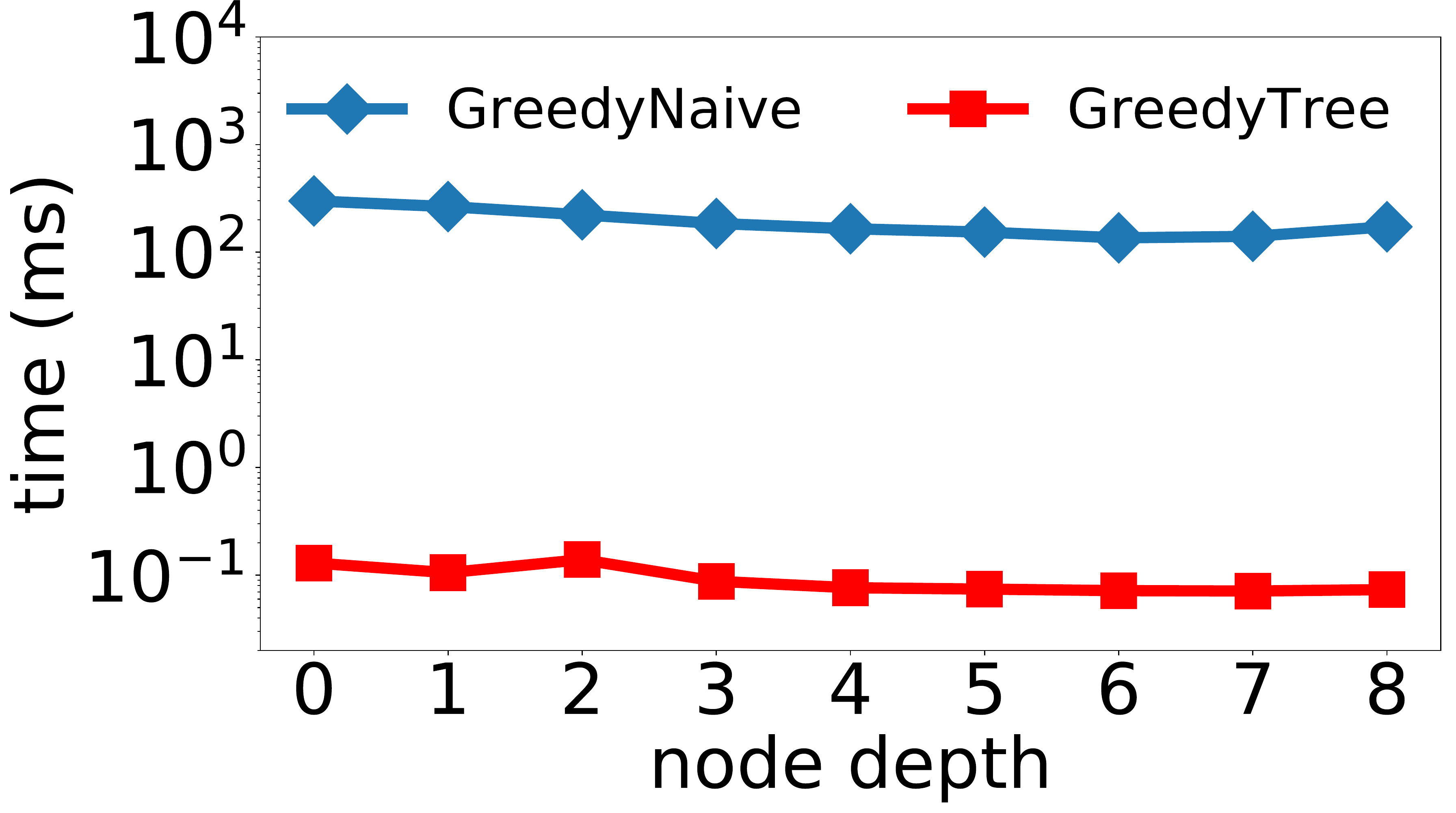}}\hfil
	\subfloat[ImageNet]{\includegraphics[width=0.243\textwidth]{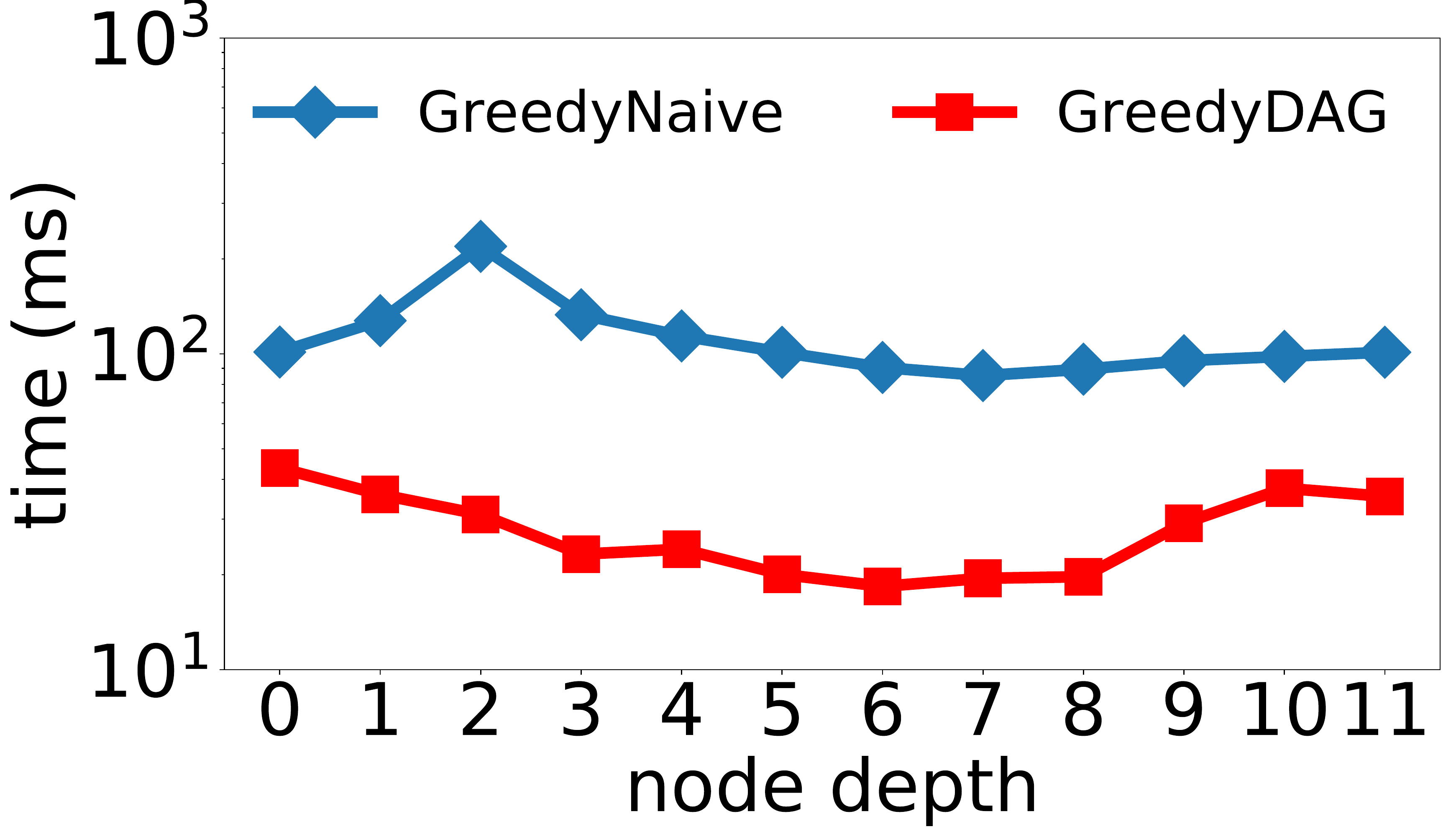}}
	\caption{Running time.}
	\label{fig:time_uniformIGS_vs_layers}
\end{figure}

\subsection{Partially Ordered Set}
From an algorithmic perspective, the IGS problem has a close relationship with the search in a partially ordered set (poset) problem. This is an extension of the well-studied problem of search in a fully ordered set~\cite{ArtOfProgramming-Searching}. Most of the researchers focus on the tree-like poset and reduce the number of comparisons in the \textit{worst-case}, which is equivalent to the WIGS~\cite{Tao:IGS} problem given the input hierarchy is a tree. It has been proved that this problem is equivalent to finding the edge ranking of a simple tree corresponding to the Hasse diagram~\cite{dereniowski2008edge}. There has been an efficient algorithm in linear time to solve this problem~\cite{ben1999optimal,onak2006generalization,mozes2008finding}. However, optimizing the worst-case cost for the general poset, i.e., WIGS~\cite{Tao:IGS} given a DAG hierarchy, is shown to be NP-hard, and there exists an $O\left(\frac{\log n}{\log\log n}\right)$-approximate polynomial-time algorithm~\cite{dereniowski2008edge}. In contrast to the problem of optimizing the worst-case cost, minimizing the average-case cost for search in poset is more challenging, which is NP-hard even if the poset has a tree structure~\cite{cicalese2011complexity}. In fact, it has been proved that there is no $o(\log n)$-approximate algorithm unless $\mathrm{NP} \subseteq \mathrm{DTIME}\left(n^{O(\log\log n)}\right)$~\cite{cicalese2011complexity}. However, surprisingly, the simple greedy method can provide a constant approximate guarantee for the tree-like poset~\cite{laber2011approximation,cicalese2014improved}.

In addition, \citet{dereniowski2017approximation} studied a problem similar to MIGS and proposed a quasi-polynomial time approximation scheme, while they tried to minimize the cost in the worst-case. \citet{bose2020competitive} further gave an online $O(\log \log n)$-competitive search tree data structure to such a MIGS-like problem.

\subsection{Decision Tree}
Moreover, as discussed in Section~\ref{sec:analysis}, AIGS is also highly relevant to decision tree. The decision tree problem~\cite{laurent1976constructing,chakaravarthy2007decision} is also known as the split tree problem~\cite{Kosaraju1999optimal} and binary identification problem~\cite{garey1972optimal} in the literature. The task is to find a scheme to unambiguously identify objects with the minimized average number of queries by a set of binary tests. The decision tree problem is a classic formulation of variant problems such as active learning, entity identification and diagnosis~\cite{dasgupta2004analysis}. We can consider the AIGS problem as a special case of the binary decision tree problem, and the search strategy of AIGS can be easily represented as a binary decision tree. It has been proven that computing the optimal solution for the decision tree problem is NP-hard~\cite{laurent1976constructing}. \citet{dasgupta2004analysis} and \citet{Kosaraju1999optimal} showed that the nature greedy algorithm has an approximation ratio of $O(\log(\frac{1}{p_{min}}))$, where $p_{min}$ is the minimum occurring probability of objects. Through some simple modifications, the rounded greedy algorithm~\cite{chakaravarthy2007decision,Kosaraju1999optimal} can easily achieve an approximation ratio of $O(\log n)$ regardless of the occurring probability.

%% file: chapters/Conclusion.tex
\section{Conclusion and Future Work}\label{sec:future_work}
In this paper, we study the problem of average-case interactive graph search (AIGS). We show that the AIGS problem is NP-hard and propose cost-effective greedy algorithms with provable theoretical guarantees, e.g.,~$(1+\sqrt{5})/2$ if the input hierarchy is a tree and $O(\log n)$ in general. Moreover, we devise efficient instantiations to accelerate the algorithms. With extensive experiments on real data, we show that our solutions considerably outperform the state of the art for AIGS.

There are some possible future research directions. For AIGS, the questions are answered by employees via a crowdsourcing platform. The employees may make mistakes during answering questions, which will introduce some noise to the search process. Existing experiment~\cite{Tao:IGS,li2020efficient} has shown that some noise is even persistent resulting from incomplete or questionable ground truth in the dataset or the subjective judgment from employees. Dealing with the negative influence of noise, especially persistent noise, is a challenge. On the theoretical side, it is interesting to derive better approximation ratios leveraging the structure of category hierarchy, such as height and maximum degree. 